\documentclass[11pt,a4paper,usanames,dvipsnames,reqno]{amsart}

\usepackage{ifthen}

\ifthenelse{\isundefined \OlafsVersion}
{ 
 \newcommand{\noTIKZ}[1]{#1} 
}
{ 
 \newcommand{\noTIKZ}[1]{} 
}
\usepackage[margin=0.8in]{geometry}

\usepackage{scrtime} 
\usepackage{lmodern} 
\usepackage{amsthm} 
\usepackage{amsmath}
\usepackage{amsfonts}
\usepackage{amssymb}
\noTIKZ{\usepackage{tikz-cd}} 
\usepackage{graphicx}
\usepackage{caption} 
\usepackage{subcaption}
\usepackage{float}
\usepackage{cleveref} 
\noTIKZ{\usetikzlibrary{arrows}}
\usepackage{mathtools} 
\usepackage{dsfont}
\usepackage{color}
\noTIKZ{\usepackage{xfrac}}
\usepackage{enumitem}

\newcommand{\G}{\mathbf G} 
\newcommand{\Gt}{\widetilde{\G}}  
\newcommand{\subG}{\mathbf H} 
\newcommand{\W}{\mathbf W} 
\newcommand{\Wt}{\widetilde{\W}}  
\newcommand{\m}{\widetilde{m}}  
\newcommand{\EG}{\mathbf G^-} 
\newcommand{\EM}{\mathbf W^-}
\newcommand{\Em}{m^-} 
\newcommand{\EA}{\alpha^-}
\newcommand\VG{\mathbf G^+} 
\newcommand{\VM}{\mathbf W^+}
\newcommand{\Vm}{m^+} 
\newcommand{\VA}{\alpha^+}

\newcommand{\D}{D} \newcommand{\SG}{\mathcal{S}}
\newcommand{\MSG}{\mathcal{MS}} 
 
\DeclareMathOperator{\Tr}{Tr} \DeclareMathOperator{\ind}{ind} 
\newcommand{\eqLapl}{{}{\Delta^{\Wt}_{\tb}(\chi)}{}}
\newcommand{\tV}{\widetilde{V}}
\newcommand{\tE}{\widetilde{E}}
\newcommand{\tb}{\widetilde{\beta}}

\newcommand{\R}{\mathbb{R}} 
\newcommand{\C}{\mathbb{C}} 
\newcommand{\Z}{\mathbb{Z}} 
\newcommand{\Torus}{\mathbb{T}}
\newcommand{\UVP}{\bigcup_{\alpha\in\mathcal{A}(\G)}}

\newtheorem{theorem}{Theorem}[section]
\newtheorem*{theorem*}{Theorem}
\newtheorem{proposition}[theorem]{Proposition}

\theoremstyle{definition} 
\newtheorem{definition}[theorem]{Definition}

\newtheorem{remark}[theorem]{Remark}

\theoremstyle{plain}

{\bf}{\it}
\theoremstyle{definition}

\numberwithin{equation}{section}
\usepackage{accents} 
\usepackage{bbm}
\newcommand{\myparagraph}[1]{\noindent\textbf{#1}}

\newcommand{\Betti}{b} 

\newcommand{\bd} {\partial} 
\DeclareMathOperator*{\dcup} {\mathaccent\cdot\cup} 

\newcommand{\e}{\mathrm e} 
\newcommand{\im}{\mathrm i} 

\newcommand{\dd} {\, \mathrm d} 

\newcommand{\orient}[1]{\accentset{\curvearrowright}{#1}}

\newcommand{\quadtext}[1]{\quad\text{#1}\quad}
\newcommand{\qquadtext}[1]{\qquad\text{#1}\qquad}

  \newcounter{lookcounter} 

\title{Covering graphs, magnetic spectral gaps and applications to polymers and nanoribbons}
\author{John Stewart Fabila-Carrasco} %
\address{Department of Mathematics, University Carlos III de Madrid,
  Avda. de la Universidad 30, 28911. Leganés (Madrid), Spain.}
\email{jfabila@math.uc3m.es}

\author{Fernando Lled\'o} %
\address{Department of Mathematics, University Carlos III de Madrid,
  Avda. de la Universidad 30, 28911. Leganés (Madrid), Spain and Instituto de Ciencias Matemáticas (CSIC-UAM-UC3M-UCM), Madrid.}
\email{flledo@math.uc3m.es}

\thanks{JSFC was supported by Spanish Ministry of Economy and
Competitiveness through project MTM2017-84098-P}
\thanks{FLl was supported by Spanish Ministry of 
Education through project DGI MTM2017-84098-P and the
\emph{Severo Ochoa} Program for Centers of Excellence in R\&D
(SEV-2015-0554).}

\begin{document}
\begin{abstract}
In this article, we analyze the spectrum of discrete magnetic Laplacians (DML) on an infinite
covering graph $\Gt \rightarrow \G=\Gt /\Gamma$ with (Abelian) lattice group $\Gamma$ and
periodic magnetic potential $\widetilde{\beta}$. We give sufficient conditions for the existence of spectral gaps in the spectrum of the DML
and study how these depend on $\widetilde{\beta}$. The magnetic potential may be interpreted as a control parameter for the spectral bands and gaps. 
We apply these results to describe the spectral band/gap structure of polymers (polyacetylene) 
and of nanoribbons in the presence of a constant magnetic field.
\end{abstract}



\maketitle

%
%
\section{Introduction}
\label{sec:intro}
%
    
It is a well known fact that the spectrum of Laplacians or, more generally, Schr\"odinger operators with
periodic potentials, on Abelian coverings have band structure. That is to say the spectrum consists of the union 
of intervals (bands) described in terms of a so-called Floquet (or Bloch) parameter which is the dual 
of the Abelian group acting on the structure. If two consecutive spectral bands of a bounded self-adjoint operator 
$T$ do not overlap, then we say that the spectrum has a spectral gap, i.e., a maximal nonempty interval $(a,b)\subset[-\|T\|,\|T\|]$ 
that does not  intersect the spectrum of the operator. This is a quite natural situation in solid state physics, where -- for
example in semiconductors or its optical counterparts, photonic crystals -- the operators modeling the 
dynamics of particles have some forbidden energy regions (see, e.g., \cite{kuc:01,kk:02}).
Depending on the type of the periodic structure involved, spectral gaps may be produced by deformation of the geometry
(cf., \cite{post:03a,lledo-post:07,lledo-post:08}) or by a suitable periodic decoration of the metric or the discrete covering graph 
(see, e.g., \cite{aizenman-schenker:00,ekw:10,koro-sabu:15,lledo-post:08b,suzuki:13} and \cite[Section~4]{kuc:05}).

In this article we study the spectrum of discrete magnetic Laplacians (DMLs for short) on infinite discrete coverings graphs 
\[
 \pi\colon \Gt \rightarrow \G=\Gt /\Gamma  \;,
\]
where $\Gamma$ is an (Abelian) lattice group acting freely and transitively on $\Gt$ (also the graph $\Gt$ is called as $\Gamma$-periodic
graph with finite quotient $\G$). We will present our analysis for graphs with arbitrary weights $m$ on vertices and arcs although the graphs presented in the examples of the last section will initially have standard weights which are more usual in the context of mathematical physics.
In addition, we consider a periodic magnetic potential $\tb$ on the arcs of the covering graph $\Gt$ 
modeling a magnetic field acting on the graph. 

We denote a weighted graph as $\W=(\G,m)$, and a magnetic
weighted graph (MW-graph for short) is a weighted graph $\W$ together with a magnetic potential acting on its arcs.
Any MW-graph $\W=(\G,m)$ with magnetic potential $\beta$ has canonically associated a DML denoted as $\Delta^{\W}_{\beta}$.
We say that $\Wt=(\Gt,\m)$ with magnetic potential $\tb$ is a $\Gamma$-periodic MW-graph if $\Gt \rightarrow \G=\Gt /\Gamma$ is a
$\Gamma$-covering and $\m$ and $\tb$ are periodic with respect to the group action.

In this article we generalize the geometric condition obtained in \cite[Theorem~4.4]{fabila-lledo-post:18}
for $\tb=0$ to non-trivial periodic magnetic potentials. In particular, if $\Wt=(\Gt,\m)$ is a $\Gamma$-periodic MW-graph with magnetic potential $\tb$, 
we will give in Theorem~\ref{teo:delta} a simple geometric condition on 
the quotient graph $\G=\Gt /\Gamma $ that guarantees the existence of non-trivial spectral gaps on the spectrum of the discrete magnetic Laplacian $\Delta_{\tb}^{\Wt}$.
To show the existence of spectral gaps, we develop a purely discrete spectral localization technique based on virtualization of arcs and vertices on
quotient $\G$. These operations produce new graphs with, in general, different weights that allow to localize the eigenvalues of the original 
Laplacian in certain intervals. We call this procedure a discrete bracketing and we refer to \cite{fabila-lledo-post:18} for additional motivation and proofs.

One of the new aspects of the present article is the generalization of results in \cite{fabila-lledo-post:18} to include
a periodic magnetic field $\tb$ on the covering graph $ \pi\colon \Gt \rightarrow \G=\Gt /\Gamma$. In this
sense, $\tb$ may be used as a control parameter for the system that serves to modify the size and the regions where 
the spectral gaps are localized. We apply our techniques to the graphs modeling the polyacetylene polymer as well as to 
graphene nanoribbons. The nanoribbons are $\Z$-periodic strips of graphene either with armchair or zig-zag boundaries. The graphic in Fig.~\ref{fig:intro}
corresponds to an armchair nanoribbon with width $3$. It can been seen 
how a periodic magnetic potential with constant value $\tb\in [0,2\pi)$ on each cycle (and plotted on the horizontal axis) affects the spectral bands 
(gray vertical intervals that appear as intersection of the region with a line $\tb=const$) and the spectral gaps (white vertical intervals).
We refer to Subsection~\ref{subsec:nanoribbon} for additional details of the construction.

\begin{figure}[H]
\includegraphics[scale=.3]{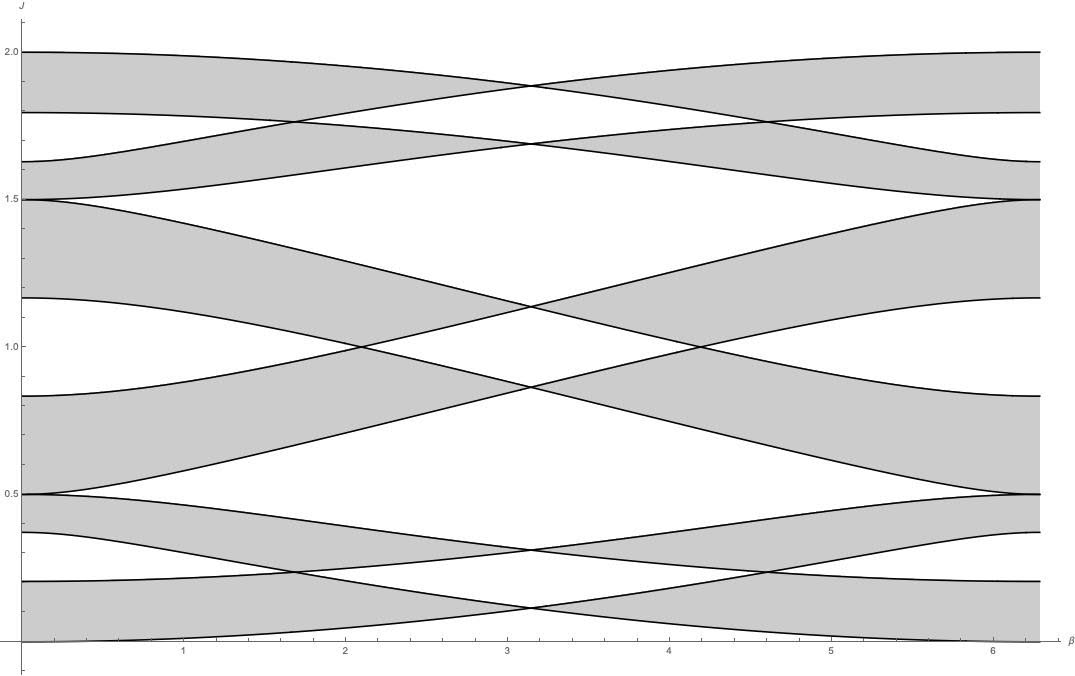}
\caption{Structure of the spectrum in spectral bands (gray) and spectral gaps (white) as a function of the constant (periodic) magnetic potential $\tb$  of \emph{3-aGNR}. }
\end{figure}\label{fig:intro}

In the case of the polyacetylene polymer we find a spectral gap that is stable under perturbation of the (constant) magnetic field. 
Moreover, if the value of the magnetic field is $\pi$ then the spectrum of the DML degenerates to four eigenvalues of infinite multiplicity
making the material almost an insulator.
 
The article is structured in five sections. In Section~\ref{sec:graphs+DML} we collect the basic definitions and results on discrete weighted multigraphs (graphs which may have loops and multiple arcs). We consider discrete magnetic potentials on the arcs and define the 
discrete magnetic Laplacian on the graph which will be the central operator in this work. In Section~\ref{sec:spectral_order} we present a spectral relation between finite MW-graphs based on an order relation between the eigenvalues of the corresponding DMLs.
Moreover, we will present the basic arc and vertex virtualization procedure that will allow one to localize the spectrum of the DML on the infinite covering graph. In Section~\ref{sec:coverings} we extend the discrete 
Floquet theory considered in \cite[Section~5]{fabila-lledo-post:18} to the case of covering graphs with periodic magnetic potentials. In Section~\ref{sec:examples} 
we apply the spectral localization results developed before to the example of $\Z$-periodic graphs modeling the polyacetylene polymer as well as 
graphene nanoribbons in the presence of a constant magnetic field.

\paragraph{\bf Acknowledgements:} It is a pleasure to thank Manuel Asorey for useful conversations
on the graphene nanoribbons example of the Section~\ref{sec:examples},
during the IWIGQMA workshop in Grajera 2019, Spain.

%
%
\section{Weighted graphs and discrete magnetic Laplacians}
\label{sec:graphs+DML}
%

In this section we introduce the basic definitions and results concerning MW-graphs and give also the definition of discrete magnetic Laplacians. 
For further motivation and results we refer to \cite{sunada:94c,lledo-post:08b,fabila-lledo-post:18} and references cited therein.

We denote by $\G=\left( V, E, \bd \right)$ a (discrete) directed multigraph which in the following 
we call simply a graph; here $V=V(\G)$ is
the set of vertices and $E=E(\G)$ the set of arcs. The orientation map is given by $\bd \colon
E\rightarrow V\times V$ and $\bd
e=\left(\bd_-e,\bd_+e \right)$ is the pair of the initial and
terminal vertices. Graphs are allowed to have multiple arcs, i.e., arcs
$e_1\neq e_2$ with $\left(\bd_-e_1,\bd_+e_1 \right) =\left(
  \bd_-e_2,\bd_+e_2 \right)$ or $\left(\bd_-e_1,\bd_+e_1 \right)
=\left( \bd_+e_2,\bd_-e_2 \right)$ as well as loops, i.e., arcs $e_1$
with $\bd_-{e_1}=\bd_+{e_1}$. Moreover, we define
\begin{equation*}
  E_v:=E_v^+ \dcup E_v^- \quadtext{(disjoint union), where}
  E^\pm_v :=\left\lbrace e\in E \mid v=\bd_\pm e\right\rbrace  \;.
\end{equation*}
With this notation the degree of a vertex is $\deg(v)=|E_v|$ and a loop increases the degree by $2$.

Given subsets $A,B \subset V$, we define 
\begin{equation*}
  E^+(A,B) := \left\lbrace e \in E \mid \bd_-e \in A, \bd_+e \in B\right\rbrace
  \quadtext{and}
  E^-(A,B) := E^+(B,A).
\end{equation*}
Moreover, we put
$E(A,B) := E^+(A,B) \cup E^-(A,B)$ and $E(A):= E(A,A)$.

To simplify the notation, we write $E(v,w)$ instead of
$E(\{v\},\{w\})$ etc.  Note that loops are not counted double in
$E(A,B)$, in particular, $E(v):=E(v,v)$ is the set of loops based at
the vertex $v\in V$.
The \emph{Betti number} $\Betti(\G)$ of a finite graph $\G=\left( V, E,
  \bd \right)$ is defined as
\begin{equation}
  \label{eq:betti}
  \Betti(\G):=|E|-|V|+1.
\end{equation}

To study the virtualization processes of
vertices, arcs and the structure of covering graphs we will need to introduce the following
substructures of a graph.

\begin{definition}
  \label{def:subgraphs}
  Let $\G=(V,E,\bd)$ be a graph and denote by
  $\subG=(V_0,E_0,\bd_0)$ a triple such that $V_0 \subset V$, $E_0
  \subset E$ and $\bd_0 =\bd \restriction_{E_0}$.
  \begin{enumerate}
  \item 
    \label{subgraphs.a}
    If $E_0 \cap E(V \setminus V_0) = \emptyset$, we say that $\subG$
    is a \emph{partial subgraph} in $\G$. We call
    \begin{align}
      \nonumber
      B(\subG,\G) 
      :=& E(V_0,V \setminus V_0)\\
      \label{eq:def.bridges}
      =&\left\lbrace e \in E 
        \mid \bd_-e \in V_0, \bd_+e \in V \setminus V_0 \text{ or }
             \bd_+e \in V_0, \bd_-e \in V \setminus V_0
      \right\rbrace
    \end{align}
    the set of \emph{connecting arcs} of the partial subgraph $\subG$
    in $\G$.
	
  \item 
      \label{subgraphs.d}
       If $E_0 \subset E(V_0)$, then $\subG$ is a \emph{subgraph} of $\G$
  \end{enumerate}
\end{definition}

Note that, in general, a partial subgraph $\subG=(V_0,E_0,\bd_0)$ is not a graph as defined above, since we may have
arcs $e \in E$ with $\bd_\pm e \notin V_0$. We do exclude though the case that $\bd_+e \notin V_0$ \emph{and} 
$\bd_-e \notin V_0$.  The arcs not mapped into $V_0 \times V_0$ under $\bd_0$ are precisely the connecting arcs of 
$\subG$ in $\G$. Partial subgraphs appear naturally as fundamental domains of covering graphs 
(cf., Section~\ref{sec:coverings}).\footnote{Note that we use the name partial subgraph in a different sense as in usual combinatorics literature.}

Let $\G=(V,E,\bd)$ be a graph; a \emph{weight} on $\G$ is a
pair of functions denoted by a unique symbol $m$ on the vertices and arcs
$m\colon V\rightarrow (0,\infty)$ and $m\colon E\rightarrow (0,\infty)$
such that $m(v)$ is the weight at the vertex $v$ and $m_e$ is the weight at $e\in E$.
We call $\W=(\G,m)$ a weighted graph.  It is natural to interpret $m$ as a positive measure
and consider $m(E_0):=\sum_{e \in E_0} m_e$ for any $E_0 \subset E$.  The
\emph{relative weight} is $\rho \colon V\rightarrow (0,\infty)$
defined as
\begin{subequations}
  \label{eq:rel.weight}
  \begin{equation}
    \label{eq:rel.weight.a}
    \rho(v)
    :=\dfrac{m(E_v)}{m(v)}
    =\dfrac{m(E_v^+)+m(E_v^-)}{m(v)}.
  \end{equation}
  In order to work with bounded discrete magnetic Laplacians 
  we will assume that the relative weight is uniformly bounded, i.e.,
  \begin{equation}
    \label{eq:rel.weight.b}
    \rho_\infty:=\sup_{v\in V} \rho(v)<\infty.
  \end{equation}
\end{subequations}

The most important and intrinsic examples of weights are 
\begin{itemize}
 \item {\em Standard weight:} $m(v)=\deg(v)$, $v\in V$, and $m_e=1$, $e\in E$, so that $\rho(v)=\rho_\infty=1$.
 \item {\em Combinatorial weight:} $m(v)= m_e=1$, $v\in V$, $e\in E$ hence $\rho(v)=\deg(v)$ and $\rho_\infty=\sup_{v\in V}\deg(v)$.
\end{itemize}

Giving a weighted graph $\W=(\G,m)$, we associate the following two natural
Hilbert spaces which we interpret as $0$-forms and $1$-forms, respectively.
\begin{align*}
  \ell_2(V,m)
  &:=\Bigl\lbrace f\colon V\rightarrow \C \mid
    \left\| f \right\|_{V,m}^2=\sum_{v \in V} | f(v)|^2 m(v) < \infty  \Bigr\rbrace
  \qquad\text{and}\\
  \ell_2(E,m)
  &:=\Bigl\lbrace \eta\colon E\rightarrow \C \mid
    \left\| \eta \right\|_{E,m}^2=\sum_{e \in E} | \eta_e|^2 m_e < \infty 
      \Bigr\rbrace,
\end{align*}
with corresponding inner products
\begin{equation*}
  \left\langle f,g\right\rangle_{\ell_2(V,m)}
  =\sum_{v \in V} {f(v)}  \overline{g(v)}  m(v)
  \quadtext{and} 
  \left\langle \eta,\zeta \right\rangle_{\ell_2(E,m)}
  =\sum_{e \in E} \eta_e  \overline{\zeta_e}  m_e \;.
\end{equation*}

Let $\G$ a graph; a \emph{magnetic potential} $\alpha$
acting on $\G$ is a $\Torus$-valued function on the arcs as follows, $\alpha\colon
E(\G)\rightarrow \Torus=\R/2\pi\Z.$ We denote the set of all vector 
potentials on $E(\G)$ just by $\mathcal{A}(\G)$.
We say that two magnetic potentials
$\alpha_1$ and $\alpha_2$ are \emph{cohomologous}, and denote this as
$\alpha_1 \sim \alpha_2$, if there is $\varphi\colon V\rightarrow
\Torus$ with
\begin{equation*}
  \alpha_1 = \alpha_2 + d\varphi.
\end{equation*}
Given a $E_0 \subset E(\G)$, we say that a magnetic potential $\alpha$
has support in $E_0$ if $\alpha_e=0$ for all $e\in E(\G)\setminus
E_0$. We call the class of weighted graphs with magnetic potential {\em MW-graphs} for short.

It can be shown that any magnetic potential on a finite graph can be
supported in $\Betti(\G)$ many arcs. 
For example, if $\G$ is a cycle, any magnetic potential is
cohomologous to a magnetic potential supported in only one arc.
Moreover, if $\G$ is a tree any magnetic potential on a tree is
cohomologous to $0$.

The \emph{twisted (discrete) derivative} is the following linear operator mapping
$0$-forms into $1$-forms:
\begin{equation}
  \label{eq:twist.der}
  d_\alpha\colon \ell_2(V,m) \rightarrow \ell_2(E,m) 
  \qquadtext{with}
  \left( d_\alpha f \right)_e
  =\e^{\im \alpha_e/2}f(\bd_+e) - \e^{-\im \alpha_e/2} f(\bd_-e).
\end{equation}  
We present next the following geometrical definition of Laplacian with magnetic field as a generalization of the discrete 
Laplace-Beltrami operator.

\begin{definition}
  Let $\W=(\G,m)$ a weighted graph with $\alpha \colon
  E \rightarrow \Torus$ a vector
  potential. The \emph{discrete magnetic Laplacian (DML for short)}
  $\Delta_\alpha  \colon
  \ell_2(V)\rightarrow \ell_2(V)$ is defined by
  $\Delta_\alpha=d_\alpha^*d_\alpha$, i.e., by
  \begin{equation*}
    \left( \Delta_\alpha f\right) \left(v \right) 
    =\rho(v)f(v)-\dfrac1{m(v)}\sum_{e\in E_v}  {\e^{\im \orient \alpha_e(v)}}f(v_e) m_e,
  \end{equation*}
  where $\orient \alpha_e(v)$ is the oriented evaluation
  and $v_e$ is the vertex opposite to $v$ along the arc $e$, i.e.,
  \begin{equation*}
    \orient \alpha_e(v)
    = 
    \begin{cases}
      -\alpha_e,  & \text{if $v=\bd_-e$,}\\
      \alpha_e,  & \text{if $v=\bd_+e$,}
    \end{cases}
    \quadtext{and}
    v_e
    = 
    \begin{cases}
      \bd_+e,  & \text{if $v=\bd_-e$,}\\
      \bd_-e   & \text{if $v=\bd_+e$.}
    \end{cases}
  \end{equation*}
  If we need to stress the dependence of the operator of the weighted graph
  $\W=(\G,m)$ we will denote the DML as $\Delta_\alpha^\W$.
\end{definition}

From this definition it follows immediately that the DML $\Delta_\alpha$ is a bounded, 
positive and self-adjoint
operator. Its spectrum satisfies $\sigma(\Delta_\alpha)\subset
[0,2\rho_\infty]$ and, in contrast to the usual Laplacian without magnetic
potential, the DML \emph{does} depend on the orientation of the graph.
If $\alpha\sim \alpha'$, then $\Delta_{\alpha}$ and $\Delta_{\alpha'}$
are unitary equivalent; in particular,
$\sigma(\Delta_\alpha)=\sigma(\Delta_{\alpha'})$. Moreover, if $\alpha\sim 0$ then $\Delta_\alpha\cong
\Delta$ where $\Delta$ denotes the usual discrete Laplacian (with vector
potential $0$). For example, if $\W=(\G,m)$ and $\G$ is a tree, then
$\Delta_\alpha^{\W} \cong \Delta^{\W}$ for any magnetic potential $\alpha$.

%
\section{Spectral ordering on finite graphs and magnetic spectral gaps}
\label{sec:spectral_order}
%

In this section we will introduce a spectral ordering relation $\preccurlyeq$ which
is invariant under unitary equivalence of the corresponding operators. Moreover, we will 
introduce two operations on the graphs (virtualization of arcs and vertices)
that will be used later to develop a spectral localization (bracketing)
of DML on finite graphs. This technique will finally be applied
to discuss the existence of spectral gaps for magnetic Laplacians on covering graphs. 
We refer to \cite{fabila-lledo-post:18,koro-sabu:15,lledo-post:08b} for additional motivation and examples. 
For proofs of the results stated in this section see \cite[Sections~3 and 4]{fabila-lledo-post:18}

Let $\W=(\G,m)$ a weighted graph.  Throughout this section, we will
assume that $|V(\G)|=n<\infty$. We denote the spectrum
of the \emph{DML}  by $\sigma(\Delta^\W_\alpha):=\{ \lambda_k(\Delta^\W_\alpha) \mid
k=1,\dots,n\}\subset [0,2\rho_\infty]$, where we will write the eigenvalues in ascending order
and repeated according to their multiplicities, i.e.,
\begin{equation*}
  0 \leq 
  \lambda_1(\Delta^\W_\alpha) 
  \leq \lambda_2(\Delta^\W_\alpha)
  \leq  \cdots \leq \lambda_n(\Delta^\W_\alpha).
\end{equation*}

\begin{definition}\label{def:PO}
  Let $\W^-$ and $\W^+$ be two finite MW-graphs of order $n^-$ and
  $n^+$, respectively, and magnetic potential $\alpha^{\pm}$. 
  Consider the eigenvalues of the DMLs $\Delta^{\W^\pm}_\alpha$ written
  in ascending order and repeated according to their
  multiplicities. 
  
\begin{enumerate}
 \item We say that $\W^-$ is \emph{spectrally smaller} than
  $\W^+$ (denoted by $\W^- \preccurlyeq \W^+$), if
  \begin{equation*}
    n^- \ge n^+
    \qquadtext{and if} 
    \lambda_k(\Delta^{\W^-}_{\alpha^-})\leq \lambda_k (\Delta^{\W^+}_{\alpha^+}) 
    \quadtext{for all}
    1\leq k \leq n^-\;,
  \end{equation*}
  where we put $\lambda_k(\Delta^{\W^+}_{\alpha^+}):=2\rho_\infty$ for $k= n^++1, \dots, n^-$ (the maximal possible eigenvalue).
 \item Consider $\W^\pm$ as above with $\W^- \preccurlyeq \W^+$. We define
  the \emph{associated $k$-th bracketing interval} $J_k=J_k(\W^-,\W^+)$
  by
  \begin{equation}
    \label{eq:brack.int}
    J_k := \bigl[\lambda_k(\Delta^{\W^-}_{\alpha^-}),\lambda_k(\Delta^{\W^+}_{\alpha^+})\bigr]
  \end{equation}
  for $k=1,\dots, n^-$.
\end{enumerate}
  
\end{definition}

Given a MW-graph, we introduce two elementary operations that consist on virtualizing arcs and vertices. 
The first one will lead to a spectrally smaller graph.

\begin{definition}[\emph{virtualizing arcs}]
  \label{deletearcs}
  Let $\W=(\G,m)$ a weighted graph with magnetic potential $\alpha$ and
  $E_0\subset E(\G)$.  We denote by $\EM=(\EG,\Em)$ the weighted
  subgraph with magnetic potential $\EA$ defined as follows:
  \begin{enumerate}
  \item $V(\EG)=V(\G)$ with $\Em(v):=m(v)$ for all $v\in V(\G)$;
  \item $E(\EG)=E(\G)\setminus E_0$ with $\Em_e:=m_e$ and
    $\bd_{\pm}^{\EG} e=\bd_{\pm}^\G e$ for all $e\in E(\EG)$;
  \item $\EA_e=\alpha_e$, $e\in E(\EG)$.
  \end{enumerate} 
  We call $\EM$ the weighted subgraph obtained from $\W$ by
  \emph{virtualizing the arcs $E_0$}. We will sometimes denote the
  weighted graph simply by $\W^-=\W-E_0$ and we write the 
  corresponding discrete magnetic Laplacian as
  $\Delta_{\EA}^{\EM}$.
\end{definition}

The second elementary operation on the graph will lead now to a spectrally larger
graph.

\begin{definition}[\emph{virtualizing vertices}]
  \label{deletevertices}
  Let $\W=(\G,m)$ a weighted graph with magnetic potential $\alpha$ and
  $V_0\subset V(\G)$.  We denote by $\VM=(\VG,\Vm)$ the weighted
  partial subgraph with magnetic potential $\VA$ defined as
  follows:
  \begin{enumerate}
  \item $V(\VG)=V(\G) \setminus V_0$ with $\Vm(v):=m(v)$ for all $v\in V(\VG)$;
  \item $E(\VG)=E(\G) \setminus \bigcup_{v_0 \in V_0}E(v_0)$ with
    $\Vm_e=m_e$ for all $e\in E(\G)$;
  \item $\VA_e=\alpha_e$, $e\in E(\VG)$.  
  \end{enumerate} 
  We call $\VM$ the weighted partial subgraph obtained from $\W$ by
  \emph{virtualizing the vertices $V_0$}. We will denote it
  simply by $\W^+=\W-V_0$.
  The corresponding discrete magnetic Laplacian is defined by 
  \begin{equation*}
    \Delta_{\VA}^{\W^+}=(d_{\VA})^*d_{\VA},
    \qquadtext{where}
    d_{\VA}:= d_\alpha \circ \iota
  \end{equation*}
  with
  \begin{equation*}
    \iota \colon \ell_2(V(\VG),\Vm) \rightarrow \ell_2(V(\G),m),
    \qquad
    (\iota f)(v)=
    \begin{cases}
      f(v), &  v \in V(\VG),\\
      0, & v \in V_0.
    \end{cases}
  \end{equation*}
\end{definition}
It can be shown that the operator $ \Delta_{\VA}^{\W^+}$ is the compression of $\Delta^\W$
onto a $(|V|-|V_0|)$-subspace.

The previous operations of arc and vertex virtualization will be used to 
localize the spectrum of intermediate DMLs. Before summarizing the technique in 
the next theorem, we need to introduce the following notion of vertex
neighborhood of a family of arcs.

\begin{definition}
  \label{def:admissible}
  Let $\G$ a graph and $E_0 \subset E(\G)$.  We say that a vertex
  subset $V_0 \subset V(\G)$ is \emph{in the neighborhood of $E_0$}
  if $E_0 \subset \bigcup_{v \in V_0} E_v$, i.e., if $\bd_+ e \in V_0$
  or $\bd_- e\in V_0$ for all $e\in E_0$.
\end{definition}
Later on $E_0$ will be the set of connecting arcs of a covering graph, and we
will choose $V_0$ to be as small as possible to guarantee the
existence of spectral gaps (this set is in general not unique).

\begin{theorem}
  \label{thm:technique}
  Let $\W=(\G,m)$ be a  finite MW-graph with magnetic potential $\alpha$ and $E_0\subset E(\G)$.  Then,
  for any subset of vertices $V_0$ in a neighborhood of $E_0$ we have
  \begin{equation}
    \W^- \preccurlyeq \W \preccurlyeq \W^+\;,
  \end{equation}
  where $\EM=(\EG,\Em)$ with $\EG=\G-E_0$ and $\VM=(\VG,\Vm)$ with
  $\VG=\G - V_0$.  In particular, we have the spectral localizing
  inclusion
  \begin{equation}
    \label{eq:loc.spec}
    \sigma(\Delta^{\W}_\alpha) 
    \subset  J
    := J(\W^-,\W^+)
    = \bigcup_{k=1}^{|V(\G)|} 
       \bigl[\lambda_k(\Delta^{\W^-}_{\alpha^-}),\lambda_k(\Delta^{\W^+}_{\alpha^+})\bigr];.
  \end{equation}
\end{theorem}

Observe that, in fact, the bracketing $J=J(\alpha)$ depends on the magnetic potential $\alpha$. In Section~\ref{sec:examples} we show
in some examples how the localization intervals $J_k$ change under the variation of the magnetic potential (see, e.g., Figure~\ref{fig:Prop}).
However, if the magnetic potential $\alpha$ has support on the virtualized arcs $E_0$, then $J$ will not depend on $\alpha$ because
$\alpha^\pm\sim 0$.

Next we make precise several notions in relation to spectral gaps that will
be need when we study covering graphs.
Recall that $\sigma(\Delta^\G_\alpha) \subset [0,2\rho_\infty]$,
where $\rho_\infty$ denotes the supremum of the relative weight,
(cf., Eq.~\eqref{eq:rel.weight}).

\begin{definition}
  \label{def:spec.gaps}
  Let $\W=(\G,m)$ be a weighted graph.
  \begin{enumerate}
  \item The \emph{spectral gaps set} of $\W$ is defined by
    \begin{equation*}
      \SG^\W
      =[0,2\rho_\infty] \setminus \sigma(\Delta^\W)
      =[0,2\rho_\infty]\cap  \rho(\Delta^\W)\;,
    \end{equation*}
    where $ \rho(\Delta^\W)$ denotes the resolvent set of the operator $\Delta^\W$.
  \item The \emph{magnetic spectral gaps set} of $\W$ is defined by
    \begin{equation*}
      \MSG^\W
      =[0,2\rho_\infty] \setminus \UVP \sigma(\Delta^\W_\alpha)
      =\bigcap_{\alpha\in\mathcal{A}(\G)} \rho(\Delta^\W_\alpha) \cap [0,2\rho_\infty].
    \end{equation*}
  \end{enumerate}
  where the union is taken over all the magnetic potential $\alpha$
  acting on $\G$.
\end{definition}


The following elementary properties follow directly from the definition:
$\MSG^\W\subset \SG^\W$.  In particular, if $\SG^\W=\emptyset$,
then $\MSG^\W=\emptyset $ or, equivalently, if $\MSG^\W \neq \emptyset$,
then $\SG^\W \neq \emptyset$.
Moreover, if $\G$ is a tree, then $\MSG^\W = \SG^\W$, as all DMLs are
unitary equivalent with the usual Laplacian $\Delta^\W$.

Up to now we have seen that arc/vertex virtualization will produce 
graphs $\W^\pm$ that allow to localize the spectrum of the DML of any intermediate MW-graph
$\W$ satisfying 
\[
 \W^-\preccurlyeq\W\preccurlyeq \W^+\;.
\]

%

\section{Periodic graphs and spectral gaps}
\label{sec:coverings}

In this section, we will study the spectrum of the DML of an infinite covering graph
with periodic magnetic potential in terms
of its Floquet decomposition. In Proposition~\ref{prp:floq.mag} we will identify the Floquet parameter of the covering
graph with a suitable set of magnetic potentials $\alpha$ on the quotient (cf., Definition~\ref{def:liftp}). This approach generalizes 
results in \cite[Section~5]{fabila-lledo-post:18} to include Laplacians on the infinite
covering graph with a periodic magnetic potential $\widetilde{\beta}$. Finally in Theorem~\ref{theo:main} we state a bracketing technique to localize the spectrum.

%

\subsection{Periodic graphs and fundamental domains}
Let $\Gamma$ be an (Abelian) lattice group and consider the $\Gamma$-covering (or $\Gamma$-periodic) graph
\begin{equation*}
  \pi\colon \Gt \rightarrow \G=\Gt /\Gamma \;.
\end{equation*}
We assume that $\Gamma$ acts freely and transitively on the connected graph $\Gt$ with finite quotient
$\G=\Gt /\Gamma$ (see also \cite[Chapters 5 and 6]{sunada:13} or \cite{sunada:08,fabila-lledo-post:18}).
This action (which we write multiplicatively) is orientation preserving, i.e., $\Gamma$ acts both on $\tV$ and $\tE$ such that
\begin{equation*}
  \bd_+(\gamma e) = \gamma (\bd_+e)
  \quadtext{and}
  \bd_-(\gamma e) = \gamma (\bd_-e)
  \qquad\text{for all $\gamma\in\Gamma$ and $e \in \tE$.}
\end{equation*}
In particular, we have $\tE_{\gamma v} =\gamma \tE_v$, $\gamma\in\Gamma$, $v\in \tV$.

In addition, we will study weighted covering graphs with a periodic weight $\m$ and periodic magnetic potential $\tb$, i.e., 
we consider $\Wt =(\Gt,\m,\tb)$ a MW-graph such that for any $\gamma\in\Gamma$ we have
\begin{equation*}
  \m(\gamma v)=\m(v) \;,\;
  v \in \tV
  \;,\;
  \m_{\gamma e}=\m_e
  \;,\;
  e \in \tE\qquadtext{and} \tb_{\gamma e}=\tb_e\;,\;
  e \in \tE\;.
\end{equation*}
Note that, by definition, the standard or combinatorial weights on a covering graph
satisfy the invariance conditions on the weights.
A $\Gamma$-covering weighted graph $\Wt =(\Gt,\m)$ naturally induces a weight $m$ and
a magnetic potential $\beta$
on the quotient graph $\G=\Gt /\Gamma$, given by $m=\m\circ \pi^{-1}$ and  $\beta=\tb\circ \pi^{-1}$.
    
We define next some useful notions in relation to covering graphs
(see, e.g., \cite[Section~5]{fabila-lledo-post:18} as well as 
\cite[Subsections~1.2 and 1.3]{koro-sabu:14} and \cite{lledo-post:08}).
\begin{definition}
  \label{def:fund.dom}
  Let $\Gt=(\tV,\tE,\widetilde{\bd})$ be a $\Gamma$-covering graph.
  \begin{enumerate}
    \item
      \label{fund.dom.a}
      A \emph{vertex}, respectively \emph{arc} \emph{fundamental domain} on a $\Gamma$-covering graph is
      given by two subsets $\D^V \subset \tV$ and $\D^E\subset \tE$
      satisfying
      \begin{align*}
        \tV&=\bigcup_{\gamma\in\Gamma}\gamma\D^V \quadtext{and}
        \gamma_1 \D^V \cap \gamma_2\D^V=\emptyset
        \quad\text{if $\gamma_1\neq \gamma_2$,}\\
        \tE&=\bigcup_{\gamma\in\Gamma}\gamma\D^E \quadtext{and}
        \gamma_1 \D^E \cap \gamma_2\D^E=\emptyset \quad\text{if
          $\gamma_1\neq \gamma_2$}
      \end{align*}
      with $\D^E \cap E(\tV \setminus \D^V)=\emptyset$ (i.e., an arc in
      $\D^E$ has at least one endpoint in $\D^V$).  We often simply
      write $\D$ for a fundamental domain, where $\D$ stands either
      for $\D^V$ or $\D^E$.
      
    \item
      \label{fund.dom.b}
      A \emph{(graph) fundamental domain} of a covering graph $\Gt$ is
      a partial subgraph (cf., Definition~\ref{def:subgraphs})
      \begin{equation*}
        \subG =(\D^V,\D^E,\bd \restriction_{\D^E}),
      \end{equation*}
      where $\D^V$ and $\D^E$ are vertex and arc fundamental
      domains, respectively. We call
      \begin{equation*}
        B(\subG,\Gt):=E(\D^V,V \setminus \D^V)
      \end{equation*}
      the set of \emph{connecting arcs} of the fundamental domain
      $\subG$ in $\Gt$.
    \end{enumerate}
\end{definition}

\begin{remark}
  \label{rem:coord}
  \indent
  \begin{enumerate}
  \item
  \label{coord.a}
  Fixing a fundamental domain on the covering graph and the group $\Gamma$ will be used to give coordinates 
  (to the arcs and vertices) on
  the covering graph $\Gt$.
 
    In fact, given a specific $\D^V$ in a $\Gamma$-covering graph $\Gt$, we can write any $v\in V(\Gt)$
    uniquely as $v=\xi(v)v_0$ for a unique pair $(\xi(v), v_0) \in
    \Gamma \times \D^V$.  This follows from the fact that the action
    is free and transitive.  We call $\xi(v)$ the
    \emph{$\Gamma$-coordinate of $v$} (with respect to the fundamental
    domain $\D^V$).  Similarly, we can define the coordinates for the
    arcs: any $e\in E(\Gt)$ can be written as $e=\xi(e)e_0$ for a
    unique pair $(\xi(e), e_0) \in \Gamma \times \D^E$.
    In particular, we have
    \begin{equation*}
      \xi(\gamma v)=\gamma \xi(v)
      \qquadtext{and}
       \xi(\gamma e)=\gamma \xi(e), \quadtext{ for all } \gamma\in\Gamma. 	
    \end{equation*}
  \item
    \label{coord.b}
    Once we have chosen a fundamental domain $\subG=(\D^V,\D^E,\bd)$, 
    we can embed $\subG$
    into the quotient $\G=\Gt/\Gamma$ of the covering $\pi \colon
    \Gt \rightarrow \G=\Gt / \Gamma$ by
    \begin{equation*}
      \D^V \rightarrow V(\G)=V/\Gamma, \quad
      v \mapsto [v]
      \qquadtext{and}
      \D^E \rightarrow E(\G)=E/\Gamma, \quad
      e \mapsto [e],
    \end{equation*}
    where $[v]$ and $[e]$ denote the $\Gamma$-orbits of $v$ and $e$,
    respectively.  By definition of a fundamental domain, these maps
    are bijective.  Moreover, if $\bd_\pm e=v$ in $\subG$, then also
    $\bd_\pm ([e]) = [v]$ in $\G$, i.e., the embedding is
    a (partial) graph homomorphism. 
  \end{enumerate}
\end{remark}

\begin{definition}
  \label{def:index}
  Let $\Gt=(\tV,\tE,\bd)$ be a $\Gamma$-covering graph with fundamental
  graph $\subG=(\D^V,\D^E,\bd)$.  We define the \emph{index} of an
  arc $e \in \tE$ as
  \begin{equation*}
    \ind_\subG(e) := \xi(\bd_+e)\left(\xi(\bd_-e)\right)^{-1} \in \Gamma.
  \end{equation*}
\end{definition}

In particular, we have $\ind_\subG \colon \tE \mapsto \Gamma$, and
$\ind_\subG(e) \ne 1_\Gamma$ iff $e \in \bigcup_{\gamma \in \Gamma}
\gamma B(\subG,\Gt)$, i.e., the index is only non-trivial on the
(translates of the) connecting arcs. Moreover, the set of indices and
its inverses generate the group $\Gamma$.

Since the index fulfils $\ind_\subG(\gamma e)=\ind_\subG(e)$ for all $\gamma
\in \Gamma$ by Remark~\ref{rem:coord}~\eqref{coord.a}, we can extend
the definition to the quotient $\G=\Gt/\Gamma$ by setting
$\ind_{\G}([e])=\ind_\subG(e) $ for all $e\in E(\G)$. We denote
also $[B(\subG,\Gt)]:=\{ [e] \mid e\in B(\subG,\Gt)\}$.

\subsection{Discrete Floquet theory}

Let $\Wt=(\tV,\tE,\widetilde{\bd},\m)$ be a weighted $\Gamma$-covering graph and
fundamental domain $\subG=(\D^V,\D^E,\bd)$ with corresponding weights
inherited from $\Wt$. In this context one has the natural Hilbert space
identifications
\begin{equation*}
 \ell_2(\tV,\m)
 \cong \ell_2(\Gamma) \otimes \ell_2(D^V,m)
 \cong \ell_2\left(\Gamma, \ell_2(\D^V,m)\right).
\end{equation*}
Floquet theory uses a partial Fourier transformation on the Abelian group that can be understood as putting coordinates on the periodic structure and allows to decompose the corresponding operators as direct integrals. Concretely, we consider
\begin{equation*}
  F \colon \ell_2(\Gamma)\rightarrow L_2(\widehat \Gamma),
  \qquad
  \left(F \textbf a\right) \left(\chi\right)
  :=\sum_{\gamma\in\Gamma} \overline{\chi(\gamma)} a_\gamma
\end{equation*}
for $\textbf a=\left\lbrace a_\gamma \right\rbrace_{\gamma\in\Gamma}
\in \ell_2(\Gamma)$ and where $\widehat \Gamma$ denotes the character group of $\Gamma$.  
We adapt to the discrete context of graphs with periodic magnetic potential $\tb$ the
main results concerning Floquet theory needed later. See, e.g., 
\cite[Section~3]{lledo-post:07} or \cite{koro-sabu:14} 
for details, additional motivation and references.

For any character $\chi\in\widehat{\Gamma}$ consider the space of \emph{equivariant functions} on vertices and arcs
\begin{eqnarray*}
  \ell_2^\chi(V,m)
  &:=&\left\lbrace g\colon V \rightarrow \C \mid
    g(\gamma v)= \chi(\gamma)g(v) 
    \text{ for all  } v\in V \text{ and } \gamma\in\Gamma\right\rbrace, \\
    \ell_2^\chi(E,m)
  &:=&\left\lbrace \eta\colon E\rightarrow \C \mid
    \eta_{\gamma e}= \chi(\gamma)\eta_e 
    \text{ for all  } e\in E \text{ and } \gamma\in\Gamma\right\rbrace.
\end{eqnarray*}
These spaces have the natural inner product defined on the fundamental domains $D^V$ and $D^E$:
\begin{equation*}
  \left\langle g_1,g_2\right\rangle :=\sum_{v\in \D^V} g_1(v) \overline{g_2(v)}  m(v)
  \quadtext{and} 
  \left\langle \eta_1,\eta_2\right\rangle :=\sum_{e\in \D^E} \eta_{1,e} \overline{\eta_{2,e}}  m_e\;.
\end{equation*}
Note that the definition of the inner
product is independent of the choice of fundamental domain (due to the
equivariance). We extend the standard decomposition to the case of 
the DML with periodic magnetic potential (see,
for example,~\cite{koro-sabu:14} or~\cite{higuchi-shirai:99}).

\begin{proposition}
  \label{prp:floq-th}
  Let $\Wt=(\Gt,\m)$  be a covering weighted graph where $\Gt=(\tV,\tE,\widetilde{\bd})$ and $\tb$ is a periodic 
  magnetic potential . Then there are unitary transformations
  \begin{eqnarray*}
    \Phi\colon\ell_2(\tV)
         &\rightarrow& \int_{\widehat \Gamma }^\oplus \ell_2^\chi(\tV,\widetilde{m}) \dd \chi
                       \quadtext{given by}\left( \Phi f \right)_\chi(v)
                       =\sum_{\gamma\in\Gamma}\overline{\chi({\gamma})} f(\gamma v) \\
   \Phi\colon\ell_2(\tE)
         &\rightarrow& \int_{\widehat \Gamma }^\oplus \ell_2^\chi(\tE,\widetilde{m}) \dd \chi
                       \quadtext{given by}\left( \Phi \eta \right)_\chi(v)
                       =\sum_{\gamma\in\Gamma}\overline{\chi({\gamma})} \eta_{\gamma e}\;,             
  \end{eqnarray*}
  such that
  \begin{equation*}
    \sigma\left(\Delta^{\Wt}_{\tb} \right) = \bigcup_{\chi \in \widehat \Gamma} \sigma \left(\eqLapl \right),
  \end{equation*}
  where equivariant Laplacian (fiber operators) are defined as $\eqLapl :=\Delta^{\Wt}_{\tb} \restriction_{\ell_2^\chi(\tV)}$.
\end{proposition}

\begin{proof}
 Consider the twisted derivative $d_{\tb}\colon \ell_2(\tV)\to \ell_2(\tE)$ specified in Eq.~(\ref{eq:twist.der}) and the equivariant
 twisted derivative on the fiber spaces defined by $d_{\tb}^\chi\colon \ell_2^\chi(\tV)\to \ell_2^\chi(\tE)$
 \[
  (d_{\tb}^\chi g)_e := e^{{i\tb_e}/{2}}g(\bd_+ e) - e^{{-i\tb_e}/{2}}g(\bd_-e)\;,\quad g\in\ell_2^\chi(\tV)\;.
 \]
It is straightforward to check that if $g \in \ell_2^\chi(\tV)$, then $d^\chi_{\tb} g \in \ell_2^\chi(\tE)$ and that
$\eqLapl=(d^\chi_{\tb})^* d^\chi_{\tb}$. Moreover, we
will show that the unitary transformations $\Phi$ intertwine these two first order operators, i.e.,
\[
 \Phi d_{\tb}f = \int_{\widehat \Gamma }^\oplus \,d_{\tb}^{\chi}\,(\Phi f)_{\chi}\dd \chi
 \;,\quad f\in\ell_2(\tV)\;.
\]
In fact, this is a consequence of the following computation that uses the invariance of the magnetic potential. 
For any $f\in\ell_2(\tV)$ and $\chi\in\widehat{\Gamma}$
\begin{eqnarray*}
   \left(\Phi\left(d_{\tb}f\right)\right)_{\chi,e} &=&  \sum_{\gamma\in\Gamma} \overline{\chi(\gamma)}\, \left(d_{\tb}f\right)_{\gamma e}
                                                     =\sum_{\gamma\in\Gamma} \overline{\chi(\gamma)}
                                                        \left[ e^{i\tb_{\gamma e}/{2}}f(\partial_+\gamma e)  -e^{-i\tb_{\gamma e}/{2}}f(\partial_-\gamma e)\right]\\
                                                   &=&  \sum_{\gamma\in\Gamma}\overline{\chi(\gamma)}
                                                        \left[ e^{i\tb_{ e}/{2}}f(\gamma\partial_+ e)  -e^{-i\tb_{e}/{2}}f(\gamma\partial_- e)\right]\\
                                                   &=&  \left(d_{\tb}^\chi \left(\Phi f \right)_{\chi}\right)_e \;.
\end{eqnarray*}
This shows that
\[
 \Delta^{\Wt}_{\tb} = \int_{\widehat \Gamma }^\oplus \; \eqLapl \dd \chi
\]
and, hence, $\sigma\left(\Delta^{\Wt}_{\tb} \right) = \mathop{\bigcup}\limits_{\chi \in \widehat \Gamma} \sigma \left(\eqLapl \right)$.
\end{proof}


\subsection{Vector potential as a Floquet parameter}

The following result shows that in the case of Abelian groups $\Gamma$
we can interpret the magnetic potential $\alpha$ on the quotient graph partially
as a Floquet parameter for the covering graph $\Gt\to\G$ (see Remark~\ref{rem:coord}~(b)).  
Moreover, recalling the definition of coordinate giving in Remark~\ref{rem:coord}~(\ref{coord.a}) we can
define the following unitary maps (see also \cite{kos:89} for a similar definition in the context 
of manifolds):
\begin{align*}
  U^V &\colon\ell_2(V,m) \to \ell_2^\chi(\tV,\m),&
  \left(U^Vf\right)( v)
  &= \chi(\xi(v))f([v]),\\
  U^E &\colon\ell_2(E,m) \to \ell_2^\chi(\tE,\m),&
  \left(U^E \eta \right)_{ e}
  &= \chi(\xi(e)) \left( \eta  \right)_{[e]}.
\end{align*}
It is straightforward to check that $U^V$ and $U^E$ are well defined and unitary.

\begin{definition}\label{def:liftp}
Let $\Gamma$-covering graph $\pi \colon \Wt \to \W$ with periodic weights $\m$, periodic
magnetic potential $\tb$ and fundamental domain $\subG$.
We denote by $\alpha$ a magnetic potential acting on the quotient $\G=\Gt/\Gamma$. We say that
$\alpha$ has \textbf{the lifting property} if there exists $\chi\in\widehat{\Gamma}$ such
that:
\begin{equation}\label{eq:liftproperty}
\e^{i\alpha_{[e]}}=\chi\left( \ind_\subG(e)\right) e^{i\tb_{[e]}}
 \quad \text{ for all } e \in E.
\end{equation}
We denote the set of all the magnetic potentials with the lifting property as 
$\mathcal{A}_\subG$.
\end{definition}

\begin{proposition}
  \label{prp:floq.mag}
  \indent
Consider a $\Gamma$-covering graph $\pi \colon \Wt \to \W$ with periodic magnetic potential $\tb$, where
$\Wt=(\Gt,\m)$, $\W=(\G,m)$ and $\subG$ is a fundamental domain. Then
\begin{equation}\label{floq.mag.a}
 \sigma(\Delta_{\tb}^{\Wt})= \bigcup_{\alpha\in \mathcal{A}_\subG} \sigma(\Delta^{\W}_{\alpha})
 \subset [0,2p_\infty] \setminus \MSG^{\W}.
\end{equation}
\end{proposition}
\begin{proof}
  By Proposition~\ref{prp:floq-th}, it is enough to show
    \begin{equation*}
     \bigcup_{\chi \in \widehat \Gamma} 
            \sigma \left(\eqLapl \right)
      = \bigcup_{\alpha\in \mathcal{A}_\subG} \sigma\left(\Delta^{\W}_{\alpha}\right)
    \end{equation*}  
To show the inclusion ``$\subset$'' consider a character $\chi \in \widehat{\Gamma}$ and define a magnetic potential on
    $\G$ as follows
    \begin{equation}
      \label{eq:rel.alpha.chi}
      \e^{\im  \alpha_{[e]}}:=\chi(\ind_\subG(e))\,e^{i\tb_{[e]}}\;,\quad e\in E\;.
    \end{equation}
  Then we have
  \begin{eqnarray*}
    \left(d_{\tb}^\chi (U^Vf)\right)_e
      &=& e^{i\tb_{e}/2} (U^Vf)(\bd_+e)-e^{-i\tb_{e}/2}(U^Vf)(\bd_-e) \\
      &=&  e^{i\tb_{e}/2}\chi(\xi(\bd_+e))f([\bd_+e]) - e^{-i\tb_{e}/2}\chi(\xi(\bd_-e))f([ \bd_-e]).
  \end{eqnarray*}
  On the other hand, we have
  \begin{equation*}
    (U^E d_\alpha f)_e
    = \chi(\xi(e))
      \Bigl(\e^{\im \alpha_{[e]}/2} f([ \bd_+e])-\e^{-\im \alpha_{[e]}/2} f([ \bd_-e])\Bigr).
  \end{equation*}
  Therefore, the intertwining equation $d_{\tb}^\chi U= U^E d_\alpha$ holds if
  \begin{equation*}
    e^{i\tb_{e}/2}\chi(\xi(\bd_+e))=\chi(\xi(e))\e^{\im \alpha_{[e]}/2}
    \quadtext{and}
    e^{-i\tb_{e}/2}\chi(\xi(\bd_-e))=\chi(\xi(e))\e^{-\im \alpha_{[e]}/2}
  \end{equation*}
  or, equivalently, if 
  \begin{equation*}
    \e^{\im \alpha_{[e]}} 
    = \chi(\xi(\bd_+e))\chi(\xi(\bd_-e))^{-1}\,e^{i\tb_{[e]}}
    = \chi(\ind_\subG(e))\,e^{i\tb_{[e]}}\;.
  \end{equation*}
  But this equation is true by definition of the magnetic potential on $\G$ given in Eq.~\eqref{eq:rel.alpha.chi}.
  Finally, since $\eqLapl=(d_{\tb}^\chi)^*d_{\tb}^\chi$ and $\Delta^{\W}_\alpha=d_\alpha^*d_\alpha$ it is clear that these Laplacians
  are unitary equivalent.

  To show the reverse inclusion ``$\supset$'' let $\alpha\in \mathcal{A}_\subG$ and $E_\subG\subset E(\G)$ is such that 
  $\{\ind_\subG(e)\mid [e]\in E_\subG\}$ is a basis of the group $\Gamma$.  Then define
    \begin{equation}
      \label{eq:rel.alpha.chi2}
      \chi(\ind_\subG(e)):=\e^{\im  \alpha_{[e]}}\,e^{-i\tb_{[e]}}\;,\quad e\in E_\subG
    \end{equation}
  and we can extend $\chi$ to all $\Gamma$ multiplicatively, so that $\chi\in \widehat{\Gamma}$. As before, we can show then
     $\sigma \left(\eqLapl \right)=  \sigma\left(\Delta^{\W}_{\alpha}\right)$ and the proof is concluded. 
\end{proof}

\subsection{Spectral localization for the DML on a covering graph}
We apply now the technique stated in Theorem~\ref{thm:technique} to covering graphs.
  
\begin{theorem}\label{theo:main}
Let $\Wt=(\Gt,\m)$ be a $\Gamma$-covering graph and $\tb$ a periodic magnetic potential.
Consider a fundamental domain $\subG=(\D^V,\D^E,\bd)$ and $\W=(\G,m)$ with magnetic potential $\beta$, where $\G=\Gt/\Gamma$.
The functions $m$ and $\beta$ are induced by $\m$ and $\tb$ respectively. Let
\begin{equation*}
  E_0 := [B(\subG,\Gt)]
\end{equation*}
be the image of the connectivity arcs on the quotient and $V_0$ in the
neighborhood of $E_0$. Define by
\begin{equation*}
  \W^- := \W - E_0
  \qquadtext{and}
  \W^+:= \W - V_0.
\end{equation*}
the corresponding arc and vertex virtualized graphs, respectively. Then
\begin{equation*}
  \sigma(\Delta^{\Wt}_{\tb}) 
  \subset  \bigcup_{k=1}^{|V(\G)|}
    \underbrace{[\lambda_k(\Delta^{\W^-}_{\beta^-}),\lambda_k(\Delta^{\W^+}_{\beta^+})]}_{=:J_k}
\end{equation*}
where the eigenvalues of $\sigma(\Delta^{\W^-}_{\beta^-})$  and $\sigma(\Delta^{\W^+}_{\beta^+})$ are written in ascending order and repeated according to their
multiplicities.
\end{theorem}
\begin{proof}
  By Proposition~\ref{prp:floq.mag} we have \begin{equation*}
    \sigma(\Delta^{\Wt}_{\tb}) = \bigcup_{\alpha\in \mathcal{A}_\subG}
    \sigma(\Delta^{\W}_{\alpha}).
    \end{equation*} 

Now, by the bracketing technique of Theorem~\ref{thm:technique}, we have for any potential with the lifting property
$\alpha\in \mathcal{A}_\subG$ (cf., Definition~\ref{def:admissible}):
\begin{equation*}
  \lambda_k(\Delta^{\W^-}_{\alpha^-}) 
  \le \lambda_k(\Delta^{\W}_{\alpha}) 
  \le \lambda_k(\Delta^{\W}_{\alpha^+})
  \quadtext{for all}
  k=1,\dots,|V(\G)|\;.
\end{equation*}
Therefore, by Eq.~(\ref{eq:loc.spec})
\begin{equation*}
  \sigma(\Delta^{\Wt}_{\tb}) = \bigcup_{\alpha\in \mathcal{A}_\subG} \sigma(\Delta^{\W}_\alpha) 
\subset   \bigcup_{\alpha\in \mathcal{A}_\subG} \bigcup_{k=1}^{|V(\G)|} 
\bigl[\lambda_k(\Delta^{\W^-}_{\alpha^-}),\lambda_k(\Delta^{\W^+}_{\alpha^+})\bigr]\; ;
\end{equation*}
since $\alpha$ has the lifting property, Eq.~(\ref{eq:liftproperty}) implies that there exists $\chi\in\widehat{\Gamma}$ such
that:
\begin{equation*}
\e^{i\alpha_{[e]}}=\chi\left( \ind_\subG(e)\right) e^{i\tb_{[e]}}
\quad \text{ for all } e \in E.
\end{equation*}
But for all $e\in E\setminus E_0= E\setminus[B(\subG,\Gt)]$ the index is trivial, i.e., $\ind_\subG(e)=1_\Gamma$ (see Remark~\ref{rem:coord}).
Thus by $\Gamma$-periodicity we obtain that $\beta_e=\alpha_e$ for all arcs $e\in E\setminus E_0$.
Since $\alpha$ and $\beta$ are magnetic potentials acting on $\G$, and $\G^-=\G-E_0$ then
then $\alpha^-=\beta^-$.
Similarly, for $\G^+=\G-V_0$ with $V_0$ in the
neighborhood of $E_0$, we have that $\alpha^+=\beta^+$. We obtain finally 
\begin{equation*}
  \sigma(\Delta^{\Wt}_{\tb}) 
  \subset  \bigcup_{\alpha\in \mathcal{A}_\subG}  \bigcup_{k=1}^{|V(\G)|}
    \underbrace{[\lambda_k(\Delta^{\W^-}_{\beta^-}),\lambda_k(\Delta^{\W^+}_{\beta^+})]}_{=:J_k}\;.
\end{equation*}
Note that the last union does not depend anymore of $\alpha$ and this fact concludes the proof.
\end{proof}

Note that the bracketing intervals $J_k$ depend on the fundamental
domain $\subG$.  A good choice is one where the set of
connecting arcs is as small as possible providing high contrast between the interior of the fundamental domain
and its boundary. In this case, we have a
good chance that the localizing intervals $J_k$ do not cover
the full interval $[0,2\rho_\infty]$. This is a discrete geometrical version of a 
``thin--thick'' decomposition as described in~\cite{lledo-post:08b}, where
a fundamental domain of the metric and discrete graph has only a few connections to its complement.

The next theorem gives a simple geometric condition on an MW-graph $\W$ for the existence of gaps in the spectrum of the DML 
on the $\Gamma$-covering graph. 
We will specify which arcs and vertices should be virtualized in $\W$ to guarantee
the existence of spectral gaps. This result generalizes the Theorem~4.4 in \cite{fabila-lledo-post:18}.
\begin{theorem}\label{teo:delta}
Let $\Wt=(\Gt,\m)$ be a $\Gamma$-covering graph with a $\Gamma$-periodic magnetic potential $\tb$. Denote
by $\W=(\G,m)$ the quotient graph with induced magnetic potential $\beta$ and induced weights $m$, respectively.

The spectrum of the DML has spectral gaps, i.e.,  $\sigma(\Delta^{\Wt}_{\tb}) \neq [0,2\rho_\infty]$, if the following condition holds:
there exists a vertex $v_0\in V(\G)$ and a fundamental domain  $\subG$ such that the connecting arcs $[B(\subG,\Gt)]$ contain no loops,
$ [B(\subG,\Gt)]\subset E_{v_0}$ and
\begin{equation}
\label{eq:weight.cond}
\delta
:= \rho(v_0) 
- \sum_{e\in [B(\subG,\Gt)]}\frac{m_e}{m((v_0)_e)} -\frac{m( [B(\subG,\Gt)])}{m(v_0)}-\lambda_1(\Delta_{\beta^-}^{W^-})>0\;,
\end{equation}
where $\rho(v_0)=m(E_{v_0}))/m(v_0)$ is the relative weight at $v_0$ and $\W^-=(\G^-,m^-)$ with
$\G^-=\G-[B(\subG,\Gt)]$.
\end{theorem}
\begin{proof}
Consider the following arc and vertex virtualized weighted graphs: 
\begin{equation*}
\W^- := \W - [B(\subG,\Gt)]
\qquadtext{and}
\W^+:= \W - \{v_0\}\;.
\end{equation*}
Then by Theorem~\ref{theo:main}, we obtain 
\begin{equation*}
\sigma(\Delta^{\Wt}_{\tb}) 
\subset  \bigcup_{k=1}^{|V(\G)|}
\underbrace{[\lambda_k(\Delta^{\W^-}_{\beta^-}),\lambda_k(\Delta^{\W^+}_{\beta^+})]}_{=:J_k}=J\subset[0,2\rho_\infty].
\end{equation*}
To prove that $\sigma(\Delta^{\Wt}_{\tb}) \neq [0,2\rho_\infty]$ it is enough to show that 
the measure of $[0,2\rho_\infty]\setminus J$ is positive and
it can be estimate from below by:
\begin{align}
\nonumber
\sum_{k=1}^{n-1}
\bigl(\lambda_{k+1}\bigl(\Delta^{\EM}_{\beta^-}\bigr)
- \lambda_k\bigl(\Delta^{\VM}_{\beta^+}\bigr)\bigr)
=& \sum_{k=2}^{n}\lambda_{k}\bigl(\Delta^{\EM}_{\beta^-}\bigr)
-\sum_{k=1}^{n-1}\lambda_k\bigl(\Delta^{\VM}_{\beta^+}\bigr)\\
\label{eq:step0}
&= \Tr\bigl(\Delta^{\EM}_{\beta^-} \bigr)-   \Tr\bigl(\Delta^{\VM}_{\beta^+}\bigr)-	\lambda_{1}\bigl(\Delta^{\EM}_{\beta^-}\bigr).          
\end{align} 
Therefore it is enough to calculate	$\Tr(\Delta^{\VG}_{\beta^+})$ and $\Tr(\Delta^{\EG}_{\beta^-})$ 
(see \cite[Proposition~3.3]{fabila-lledo-post:18}).

\myparagraph{Step 1: Trace of $\Delta^{\EG}_{\beta^-}$.}  We define
$\EM=(\EG,\Em)$ where $\G^-=\G-[B(\subG,\Gt)]$. Recall that
$V(\G^-)=V(\G)$, $E(\G^-)=E(\G)\setminus [B(\subG,\Gt)]$; the
weights on $V(\G^-)$ and $E(\G^-)$ coincide with the
corresponding weights on $\W$.  The relative weights of $\EM$ are
\begin{equation*}
\rho^-(v)=
	\begin{cases}
	\rho^{\W}(v)-\dfrac{
		m \left( [B(\subG,\Gt)]\right) }{m(v)}, &\text{if $v=v_0$,}\\[2ex]
	\rho^{\W}(v)-\dfrac{m\left(  [B(\subG,\Gt)]\cap E_v\right) }{m(v)}, 
	&\text{if $v\in B_{v_0}$,}\\[2ex]
	\rho^{\W}(v), &  \text{otherwise,}
	\end{cases}
\end{equation*}
where 
\begin{equation*}
	B_{v_0}
	= \{v \in V(\G) \mid 
	v=(v_0)_e \text{ for some } e\in [B(\subG,\Gt)] \text{ with }
	v\neq v_0 \}.
\end{equation*}
The trace of $\Delta^{\EG}_{\beta^-}$ is now
	\begin{align}
	\nonumber
	\Tr\bigl(\Delta^{\EG}_{\beta^-}\bigr)&=\sum_{k=1}^n \lambda_k\bigl(\Delta^{\EG}_{\beta^-}\bigr)
	=\sum_{v\in V(\G)} \rho^-(v)\\
	\label{eq:step1}
	&= \sum_{v\in V(\G)}\rho^{\W}(v)-\dfrac{
		m \left( [B(\subG,\Gt)]\right) }{m(v)}-
	\sum\limits_{v\in B_{v_0}}\dfrac{m\left( [B(\subG,\Gt)]\cap E_v\right) }{m(v)}.
	\end{align}
	
\myparagraph{Step 2: Trace of $\Delta^{\VM}$.} Let $\W^+=\left(\VG,m^+ \right) $, then the trace of
$\Delta^{\VM}_{\beta^+}$ is given by
	\begin{equation}
	\label{eq:step2}
	\Tr\bigl(\Delta^{\VM}_{\beta^+})
	=\sum_{k=1}^{n-1} \lambda_k\bigl(\Delta^{\W}_{\beta}\bigr)
	=\sum_{\substack{v\in V(\G) \\v\neq v_0}} \rho^{\W}(v).
	\end{equation}
Combining Eqs.~\eqref{eq:step0}, \eqref{eq:step1} and \eqref{eq:step2} we obtain
\begin{align*}
 \Tr\bigl(\Delta^{\EM}_{\beta^-} \bigr)-   \Tr\bigl(\Delta^{\VM}_{\beta^+}\bigr)-	\lambda_{1}\bigl(\Delta^{\EM}_{\beta^-}\bigr)
& =\rho^{\W}(v_0)
-\frac{ m([B(\subG,\Gt)])}{m(v_0)}
-\sum_{v\in B_{v_0}} \frac {m([B(\subG,\Gt)]\cap E_v)}{m(v)}-	\lambda_{1}\bigl(\Delta^{\EM}_{\beta^-}\bigr)\\
& =\rho^{\W}(v_0)
-\frac{ m \left( [B(\subG,\Gt)]\right)}{m(v_0)}
- \sum_{e\in [B(\subG,\Gt)]}\frac{m_e}{m((v_0)_e)} -	\lambda_{1}\bigl(\Delta^{\EM}_{\beta^-}\bigr)=\delta   
\end{align*}
as defined in Eq.~\eqref{eq:weight.cond}.  This shows that if $\delta>0$, then the spectrum of the 
DML is not the full interval.
	
\end{proof}

\begin{remark}
	\label{rem:gen.weights}
	\indent
	\begin{enumerate}
  \item If the graph has the \emph{standard weights}, the condition
becomes:
\begin{equation}
\delta= 1 - \sum_{e\in [B(\subG,\Gt)]}\dfrac{1}{\deg((v_0)_e)}-
\dfrac{|[B(\subG,\Gt)]|}{\deg(v_0)}-\lambda_1(\Delta_{\beta^-}^{W^-})>0\;,
\end{equation}
where $|[B(\subG,\Gt)]|$ denote the cardinality of the set $[B(\subG,\Gt)]$.

  \item If we have the \emph{combinatorial weights}, the
condition becomes simply:
\begin{equation}\label{eq:conditioncom}
\delta= \deg(v_0) - 2\;|[B(\subG,\Gt)]|-\lambda_1(\Delta_{\beta^-}^{W^-})>0\;.
\end{equation}  

 \end{enumerate}
\end{remark}
%
%

%
%
\section{Examples.}
\label{sec:examples}
In this final section, we show some examples of graphs with standard weights that are used as models of important
chemical compounds, as the polyacetylene and the graphene nanoribbons. We use the bracketing
technique developed before to localize the spectrum and the gaps of these infinite covering graphs under the
action of a periodic magnetic potential $\tb$. In particular, we will show in these examples the dependence of the spectral gaps on $\tb$.

Given $\Wt=(\Gt,\m)$ a periodic weighted graph, we consider for simplicity in this section only periodic magnetic potentials $\tb$ with the property
that the flux through all cycle on $\Gt$ is constant and equal to $s$ for some $s\in[0,2\pi)$. Since two magnetic potentials are cohomologous iff they induce the same flux through all the cycles on the graph, all periodic magnetic potentials are constant in this sense and are determinate by the value $s$. 
That is to say even if $\tb$ is a function on the arcs, we can identify it
with one value in $\Torus$. We call this choice a constant magnetic field. A similar analysis can be done for non-constant magnetic potentials.


\subsection{Polyacetylene with magnetic field}
\label{subsec:poly}
%
For the first illustration of the existence of spectral gaps for covering graphs
with periodic magnetic potential, we study the graph modeling polyacetylene, an organic polymer that consists of a chain of carbon atoms (white circles) with alternating single 
and double bonds between them, each with one hydrogen atoms (black vertex).
We denote this MW-graph as $\Wt=(\Gt,\m)$, where $\Gt$ is in Figure~\ref{subfig:PA} 
and $\m$ are the standard weights. 
The polyacetylene belongs to the family of polymers, a chemical compound in long repeated
chains that can be naturally modeled by covering graphs. The polymers have important electrical properties (see, e.g., \cite{chien2012,shirakawa:01} and references therein).
In particular, the polyacetylene is a simple polymer with good electric conductance (cf., \cite{ekn:83}).
In \cite{fabila-lledo-post:18} we study the spectrum of the Laplacian in the infinite polyacetylene graph without any magnetic field.
Applying the results of the Section~\ref{sec:coverings}, we can now study the spectrum of the DML in the polyacetylene graph under the action
of a periodic magnetic potential, in particular, the size and localization of the spectral gaps.
For the polyacetylene we will prove the next facts:
\begin{itemize}
	\item \textit{Fact 1.} Let $\m$ be the standard weights and $\tb$ a constant periodic magnetic potential. We show how to apply the
	bracketing technique to localize the spectrum for a specific value of the magnetic potential (equal to $\pi/2$) and then
	we show how the bracketing intervals change as a function of $\tb$. We will show the existence of spectral gaps.
	\item \textit{Fact 2.} Let $\m$ be the combinatorial weights and $\tb$ a periodic magnetic potential (not necessarily constant).
	Using the condition on $\delta$ in Eq.~\ref{eq:conditioncom} we show the existence of spectral gaps.
	\item \textit{Fact 3.} Let $\m$ be the standard weights, we show the existence of periodic magnetic spectral gaps, i.e.,
	a spectral gap which is stable under any perturbation by the constant periodic magnetic field. 
\end{itemize}

\textit{Fact 1.} We define a periodic magnetic potential $\tb$ acting as in Figure~\ref{subfig:PA}, i.e., the potential acts only on the cycles 
defined by the double bonds.
Observe that the action of any constant magnetic field on the polymer can be represented by putting
a suitable value $s$ for the magnetic potential as in Figure~\ref{subfig:PA}. To be concrete, we put first the value $s=\pi/2$
and want to specify the band/gap structure of the spectrum $\sigma(\Delta_{\tb}^{\Wt})$.
The graph $\Gt$ in Figure~\ref{subfig:PA} is the infinite covering of the finite graph $\G$ in Figure~\ref{subfig:PA_0}.
This graph is bipartite and has Betti number 2. In this case, if $\W=(\G,m)$ with $m$ the standard weights, we have
by Proposition~\ref{prp:floq.mag} that
\begin{equation*}
\sigma(\Delta_{\tb}^{\Wt}) = \bigcup_{t\in [0,2\pi)}
\sigma(\Delta^{\W}_{\alpha^t}),
\end{equation*}
where $\alpha^t$ is a magnetic potential acting on the quotient $\W$ with $\alpha^t(e_1)=t$, $\alpha^t(e_2)=s$
and zero in all the other arcs. Define
$E_0:=\{e_1\}$ and $V_0:=\{v_1\}$, so that $V_0$ is in the neighborhood
of $E_0$ (see Definition \ref{def:admissible}). Then we construct $\W^+$ 
and $\W^-$ as before virtualizing arcs and vertices, i.e., $\G^-:=\G-E_0$ and $\G^+:=\G-V_0$ as in Figure~\ref{subfig:PA_1}.
The induced weights $m^-$ is defined as in Definition~\ref{deletearcs} and $m^+$ as in \ref{deletevertices}. Using the notation of the
Theorem~\ref{thm:technique} and Theorem~\ref{theo:main}
we get $\sigma(\Delta^{\Gt}_{\tb})\subset J\subset [0,2]$, where $J$ is the union of the localizing intervals $J_k$
(see Figure~\ref{subfig:spectra} for the case of $s=\pi/2$). 
Since $\G$ is bipartite
we have the symmetry of spectrum under the function $\kappa(\lambda)=2-\lambda$ (cf., \cite[Proposition~2.3]{lledo-post:08b}), hence
we also have $\sigma(\Delta^{\Gt}_{\tb})\subset \kappa(J)$. Therefore, the intersection gives a
finer localization of the spectrum, i.e., we obtain finally $\sigma(\Delta^{\Gt}_{\tb})\subset
J\cap \kappa(J)$. In this example our method works almost
perfectly, since we are able to determine almost precisely the spectrum:
\begin{equation*}
J\cap\kappa(J)\setminus{\{1\}}=\sigma\left(\Delta^{\Wt}_{\tb}\right).
\end{equation*}
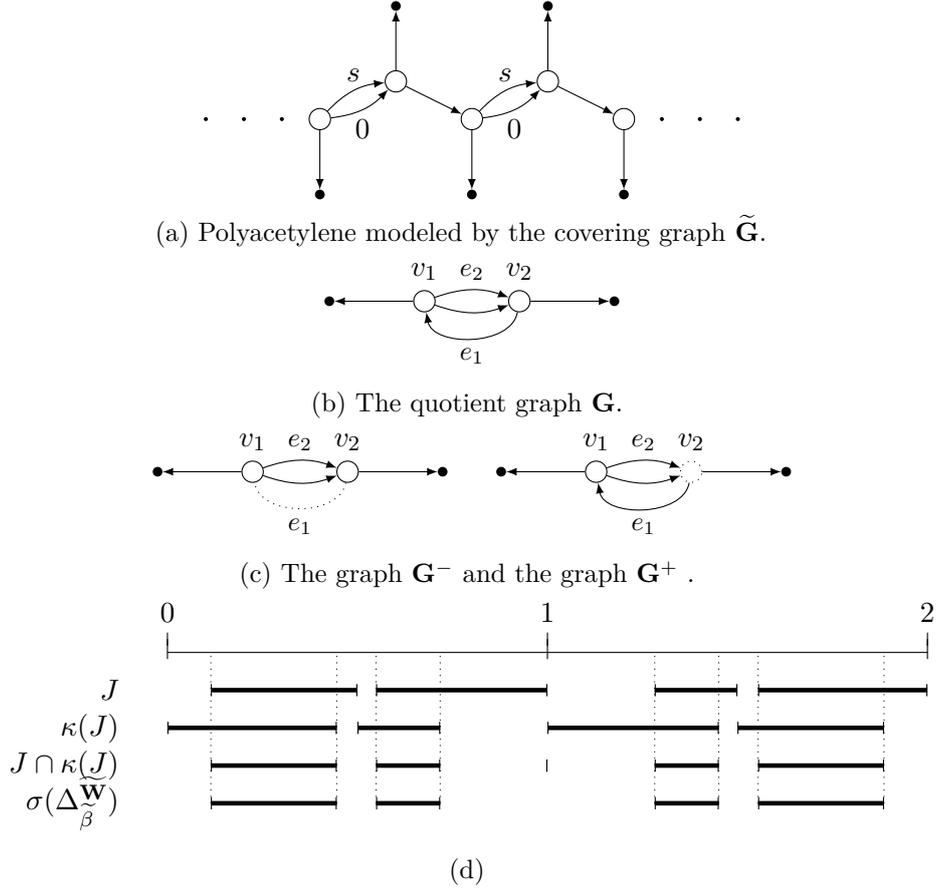
\begin{figure}[h]
	\centering 
	\subcaptionbox{Polyacetylene modeled by the covering graph $\Gt$.
		\label{subfig:PA}}%
	[.6\linewidth]{ 
		\begin{tikzpicture}[auto,
		vertex/.style={circle,draw=black!100,fill=black!100, thick,
			inner sep=0pt,minimum size=1mm},scale=2]
		\node (D) at (-1.25,0) [vertex,inner sep=.25pt,minimum
		size=.25pt,label=above:]{}; %
		\node (D) at (-1,0) [vertex,inner sep=.25pt,minimum
		size=.25pt,label=above:]{}; %
		\node (D) at (-.75,0) [vertex,inner sep=.25pt,minimum
		size=.25pt,label=above:]{}; %
		\node (O) at (-.5,0) [circle, minimum width=8pt, draw, inner
		sep=0pt] {}; %
		\node (A) at (0,.25) [circle, minimum width=8pt, draw, inner
		sep=0pt] {}; %
		\node (B) at (.5,0) [circle, minimum width=8pt, draw, inner
		sep=0pt]{}; %
		\node (C) at (1,.25) [circle, minimum width=8pt, draw, inner
		sep=0pt]{}; %
		\node (D) at (1.5,0) [circle, minimum width=8pt, draw, inner
		sep=0pt]{}; %
		\draw[-latex] (O) to[bend left=20] node[above] {$s$}(A); %
		\draw[-latex] (O) to[bend left=-20] node[below] {$0$} (A); %
		\draw[-latex] (A) to[] node[above] {} (B); %
		\draw[-latex] (B) to[bend left=20] node[above] {$s$} (C); %
		\draw[-latex] (B) to[bend left=-20] node[below] {$0$} (C); %
		\draw[-latex] (C) to[] node[above] {} (D); %
		\node (O1) at (-.5,-.5) [vertex,label=above:]{}; %
		\node (A1) at (0,.75) [vertex,label=above:]{}; %
		\node (B1) at (.5,-0.5) [vertex,label=above:]{}; %
		\node (C1) at (1,.75) [vertex,label=above:]{}; %
		\node (D1) at (1.5,-.5) [vertex,label=above:]{}; %
		\draw[-latex] (O) to[] node[above] {} (O1); %
		\draw[-latex] (A) to[] node[above] {} (A1); %
		\draw[-latex] (B) to[] node[above] {} (B1); %
		\draw[-latex] (C) to[] node[above] {} (C1); %
		\draw[-latex] (D) to[] node[above] {} (D1); %
		\node (D) at (1.75,0) [vertex,inner sep=.25pt,minimum
		size=.25pt,label=above:]{}; %
		\node (D) at (2,0) [vertex,inner sep=.25pt,minimum
		size=.25pt,label=above:]{}; %
		\node (D) at (2.25,0) [vertex,inner sep=.25pt,minimum
		size=.25pt,label=above:]{}; %
		\end{tikzpicture}}
	
	\subcaptionbox{The quotient graph $\G$.\label{subfig:PA_0}}
	[.3\linewidth]{
		\begin{tikzpicture}[auto,
		vertex/.style={circle,draw=black!100,fill=black!100, thick,
			inner sep=0pt,minimum size=1mm},scale=2.5]
		\node (O) at (-.5,0) [circle, minimum width=8pt, draw, inner
		sep=0pt, label=above:$v_1$] {}; %
		\node (A) at (0,0) [circle, minimum width=8pt, draw, inner
		sep=0pt, label=above:$v_2$] {}; %
		\draw[-latex] (O) to[bend left=20] node[above] {\small$e_2$} (A); %
		\draw[-latex] (O) to[bend left=-20] node[below] {\small{}} (A); %
		\draw[-latex] (A) to[bend left=80] node[below] {\small$e_1$} (O); %
		\node (O1) at (-1,0) [vertex,label=above:]{}; %
		\node (A1) at (.5,0) [vertex,label=above:]{}; %
		\draw[-latex] (O) to[bend left=0] node[above] {} (O1); %
		\draw[-latex] (A) to[bend left=-0] node[above] {} (A1); %
		\end{tikzpicture}}\\

	\subcaptionbox{The graph $\G^-$ and the graph $\G^+$ .\label{subfig:PA_1}}
	[.9\linewidth]{
		\begin{tikzpicture}[auto,
		vertex/.style={circle,draw=black!100,fill=black!100, thick,
			inner sep=0pt,minimum size=1mm},scale=2.5]
		\node (O) at (-.5,0) [circle, minimum width=8pt, draw, inner
		sep=0pt, label=above:$v_1$] {}; %
		\node (A) at (0,0) [circle, minimum width=8pt, draw, inner
		sep=0pt, label=above:$v_2$] {}; %
		\draw[-latex] (O) to[bend left=20] node[above] {\small$e_2$} (A); %
		\draw[-latex] (O) to[bend left=-20] node[below] {\small{}} (A); %
		\draw[dotted] (A) to[bend left=80] node[below] {\small$e_1$} (O); %
		\node (O1) at (-1,0) [vertex,label=above:]{}; %
		\node (A1) at (.5,0) [vertex,label=above:]{}; %
		\draw[-latex] (O) to[bend left=0] node[above] {} (O1); %
		\draw[-latex] (A) to[bend left=-0] node[above] {} (A1); %
		\end{tikzpicture}	\quad \begin{tikzpicture}[auto,
		vertex/.style={circle,draw=black!100,fill=black!100, thick,
			inner sep=0pt,minimum size=1mm},scale=2.5]
		\node (O) at (-.5,0) [circle, minimum width=8pt, draw, inner
		sep=0pt, label=above:$v_1$] {}; %
		\node (A) at (0,0) [circle, dotted,minimum width=8pt, draw, inner
		sep=0pt, label=above:$v_2$] {}; %
		\draw[-latex] (O) to[bend left=20] node[above] {\small$e_2$} (A); %
		\draw[-latex] (O) to[bend left=-20] node[below] {\small{}} (A); %
		\draw[-latex] (A) to[bend left=80] node[below] {\small$e_1$} (O); %
		\node (O1) at (-1,0) [vertex,label=above:]{}; %
		\node (A1) at (.5,0) [vertex,label=above:]{}; %
		\draw[-latex] (O) to[bend left=0] node[above] {} (O1); %
		\draw[-latex] (A) to[bend left=-0] node[above] {} (A1); %
		\end{tikzpicture}}\\
	\centering
	\subcaptionbox{\label{subfig:spectra}}
	[\linewidth]{	
		\begin{tikzpicture}[scale=5]
		\draw[-] (0,.5) -- (2,0.5) ; 
		\foreach \x in {0,1,2} 
		\draw[shift={(\x,0.5)},color=black] (0pt,1pt) -- (0pt,-.5pt); %
		\foreach \x in {0,1,2} 
		\draw[shift={(\x,0.5)},color=black] (0pt,0.5pt) -- (0pt,1.5pt)
		node[above] {$\x$}; %
		
		\draw[] (-.1,.4)node[left] {$J$}; %
		\draw[|-|] (0.114212,.4) -- (0.5,.4); %
		\draw[line width=.6mm] (0.114212,.4) -- (0.5,.4); %
		\draw[|-|] (0.549103,.4) -- (1,.4); %
		\draw[line width=.6mm] (0.549103,.4) -- (1,.4); %
		\draw[|-|] (1.28223,.4) -- (1.5,.4); %
		\draw[line width=.6mm] (1.28223,.4) -- (1.5,.4); %
		\draw[|-|] (1.55445,.4) -- (2,.4); %
		\draw[line width=.6mm] (1.55445,.4) -- (2,.4); %
		
		\draw[] (-.1,.3)node[left] {$\kappa(J)$}; %
		\draw[|-|] (0,.3) -- (0.44555,.3); %
		\draw[line width=.6mm] (0,.3) -- (0.44555,.3); %
		\draw[line width=.6mm] (0.5,.3) -- (0.717765,.3); %
		\draw[|-|] (0.5,.3) -- (0.717765,.3); %
		\draw[line width=.6mm] (1,.3) -- (1.4509,.3); %
		\draw[|-|] (1,.3) -- (1.4509,.3); %
		\draw[line width=.6mm] (1.5,.3) -- (1.88579,.3); %
		\draw[|-|] (1.5,.3) -- (1.88579,.3); %
		
		\draw[] (-.1,.2)node[left] {$J\cap\kappa(J)$}; %
		\draw[|-|] (0.114212,.2) -- (0.44555,.2); %
		\draw[line width=.6mm] (0.114212,.2) -- (0.44555,.2); %
		\draw[|-|] (0.549103,.2) -- (0.717765,.2); %
		\draw[line width=.6mm] (0.549103,.2) -- (0.717765,.2); %
		\draw[|-|] (.9999,.2) -- (1.00001,.2); %
		\draw[|-|] (2-0.114212,.2) -- (2-0.44555,.2); %
		\draw[line width=.6mm] (2-0.114212,.2) -- (2-0.44555,.2); %
		\draw[|-|] (2-0.549103,.2) -- (2-0.717765,.2); %
		\draw[line width=.6mm] (2-0.549103,.2) -- (2-0.717765,.2); %

		\draw[] (-.1,.1)node[left] {$\sigma(\Delta^{\Wt}_{\tb})$}; %
		\draw[|-|] (0.114212,.1) -- (0.44555,.1); %
		\draw[line width=.6mm] (0.114212,.1) -- (0.44555,.1); %
		\draw[|-|] (0.549103,.1) -- (0.717765,.1); %
		\draw[line width=.6mm] (0.549103,.1) -- (0.717765,.1); %
		
		\draw[|-|] (2-0.114212,.1) -- (2-0.44555,.1); %
		\draw[line width=.6mm] (2-0.114212,.1) -- (2-0.44555,.1); %
		\draw[|-|] (2-0.549103,.1) -- (2-0.717765,.1); %
		\draw[line width=.6mm] (2-0.549103,.1) -- (2-0.717765,.1); %

		\draw[dotted] (0.44555,.1) -- (0.44555,.5); %
		\draw[dotted] (0.114212,.1) -- (0.114212,.5); %
		\draw[dotted] (0.717765,.1) -- (0.717765,.5); %
		\draw[dotted] (0.549103,.1) -- (0.549103,.5); %
		\draw[dotted] (2-0.44555,.1) -- (2-0.44555,.5); %
		\draw[dotted] (2-0.114212,.1) -- (2-0.114212,.5); %
		\draw[dotted] (2-0.717765,.1) -- (2-0.717765,.5); %
		\draw[dotted] (2-0.549103,.1) -- (2-0.549103,.5); %
		\end{tikzpicture}}                  	      	     	           	
	\caption{Spectral gaps of the \emph{polyacetylene} graph for a constant magnetic potential
		$\beta=s$. Here,
		$J$ is the spectral localization for the pair $\G - \{e_1\}$ and $\G
		-\{v_1\}$. Bipartiteness gives a finer localization $J \cap \kappa(J)$.
		In this case we obtain the spectrum almost exactly, except for the spectral value
		$1$.}
	\label{fig:acetylene}
\end{figure}    

In conclusion, given a covering graph $\Wt$ with a periodic magnetic potential $\tb$ (see Figure~\ref{fig:acetylene} for $s=\pi/2$),
we were able to almost determine $\sigma(\Delta^{\Wt}_{\tb})$ just by specifying the localization of the spectrum given by $J\cap\kappa(J)$ (and without
computing explicitly the spectrum). 
Obviously, $J$ depends on $\tb$ and therefore of the value of $s$. Therefore for each value of $\tb$ we can construct a bracketing $J(\tb)$ of intervals for the spectrum of $\Delta^{\Wt}_{\tb}$ and, since in this case we have the reflection symmetry specified by $\kappa$ and an additional interlacing property of $\W^-$ due to 
Cauchy's theorem (to be shown in \cite{fabila-lledo-post:pre19}) we are able to give a much finer localization of the spectrum. 
In Figure~\ref{fig:Prop} we plot the spectrum $\sigma(\Delta^{\Gt}_{\tb})$ of the DML as a function of the periodic magnetic potential
$\tb$ varying within the interval $[0,2\pi]$. Here one can appreciate how the size of the gaps and their localization within the interval
$[0,2]$ changes as a function of the external magnetic field.

 \begin{figure}
 \includegraphics[width=0.7\linewidth]{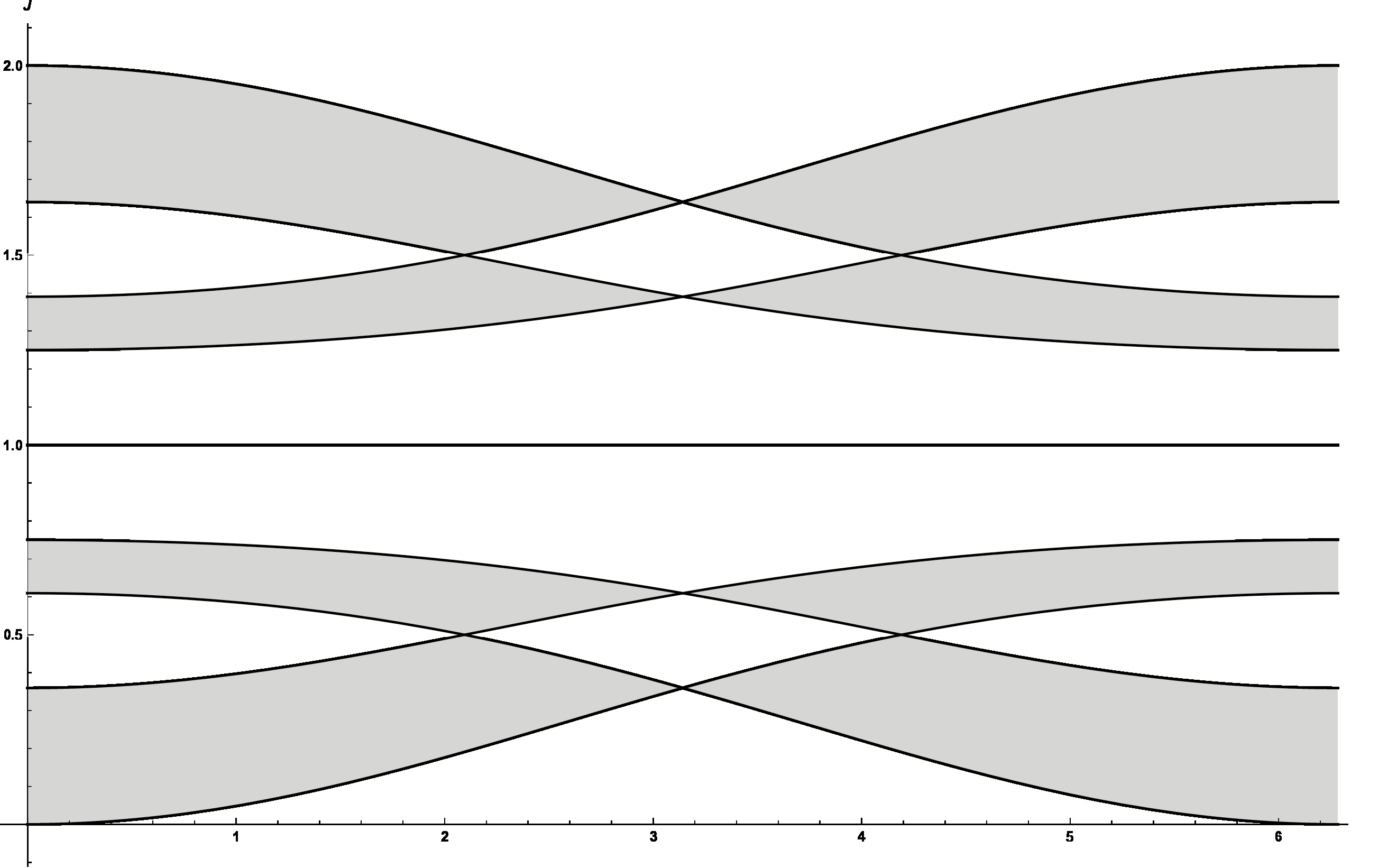}
 \caption{The horizontal axis represents the values of the magnetic potential $\tb\in[0,2\pi)$ acting on the polyacetylene polymer with standard weights. 
 For any fixed $\tb$ we obtain the intervals $J$ given by the bracketing technique as we did in the case $\tb=\pi/2$ in the Figure~\ref{fig:acetylene} 
 (and also using the symmetry given by the bipartiteness). In the vertical axis we represent the spectral bands and gaps
 for each constant value $\tb$.}\label{fig:Prop}
 \end{figure}
 
\textit{Fact 2.} We have proved using the bracketing technique that the polyacetylene with standard weights has spectral gaps for any constant periodic magnetic potential acting on it. Now, if we consider the polyacetylene with combinatorial weight, we will prove more easily the existence of spectral gaps for
all periodic magnetic potentials (not necessarily constant). Formally, let $\W=(\Gt,\m)$ be the MW-graph where $\Gt$ is the polyacetylene (Figure~\ref{subfig:PA}), $\m$ are the combinatorial weights and $\tb$ any periodic magnetic potential. Let $\G^-$ as in \textit{Fact 1}, but now
$m^-$ are also the combinatorial weights. First, we observe that $\lambda_1(\Delta_{\beta^-}^{W^-})<2$,  then we calculate $\delta$ from condition in Eq.~\ref{eq:conditioncom},
i.e., \begin{equation*}
\delta= \deg(v_1) - 2\;|[B(\subG,\Gt)]|-\lambda_1(\Delta_{\beta^-}^{W^-})>4-2-2=0,
\end{equation*}  
then by Theorem~\ref{teo:delta} we have spectral gaps. Observe we do this without compute explicitly any eigenvalue.

\textit{Fact 3.} Our method of virtualizing suitable arcs and vertices allows to proceed also alternatively. Define now
$E_1:=\{e_1,e_2\}$ and $V_1:=\{v_1\}$ so that $V_1$ is a neighborhood
of $E_1$ (see Definition \ref{def:admissible}). We construct as usual the MW-graphs $\W_1^+$ 
and $\W_1^-$ setting $\G_1^+=\G-E_1$ and $\G_1^-=\G-V_1$ as in Figure~\ref{fig:plotspectral}
and inducing the weights as in Definition~\ref{deletearcs} and~\ref{deletevertices} (observe that in this
case $\W^+_1=\W^+$). Using the notation of the
Theorem~\ref{thm:technique} and Proposition~\ref{prp:floq.mag} we observe now that the spectral localization intervals 
do {\em not} depend on the periodic magnetic potential. In fact, using the same idea that before we obtain

\begin{equation*}
\sigma(\Delta^{\Wt}_{\tb})\subset [0,3/4]\cup[5/4,2] \quadtext{for all periodic constant magnetic potential} \tb,
\end{equation*}
in particular, $(3/4,5/4)$ is a spectral gap which is stable under any perturbation by the magnetic field. 
Finally, we note that if the magnetic potential has a constant value equal to
$\pi$ then the spectrum degenerates to four eigenvalues with infinity multiplicity, i.e., the gaps
consist of the whole interval $[0,2]$ except for the four eigenvalues. In this case, the polyacetylene
becomes essentially an insulator under the influence of this particular value of the magnetic field.
\begin{figure}
	{\begin{tikzpicture}[auto,
		vertex/.style={circle,draw=black!100,fill=black!100, thick,
			inner sep=0pt,minimum size=1mm},scale=2.5]
		\node (O) at (-.5,0) [circle, minimum width=8pt, draw, inner
		sep=0pt, label=above:$v_1$] {}; %
		\node (A) at (0,0) [circle, minimum width=8pt, draw, inner
		sep=0pt, label=above:$v_2$] {}; %
		\draw[dotted] (O) to[bend left=20] node[above] {\small$e_2$} (A); %
		\draw[-latex] (O) to[bend left=-20] node[below] {\small{}} (A); %
		\draw[dotted] (A) to[bend left=80] node[below] {\small$e_1$} (O); %
		\node (O1) at (-1,0) [vertex,label=above:]{}; %
		\node (A1) at (.5,0) [vertex,label=above:]{}; %
		\draw[-latex] (O) to[bend left=0] node[above] {} (O1); %
		\draw[-latex] (A) to[bend left=-0] node[above] {} (A1); %
		\end{tikzpicture}\quad \begin{tikzpicture}[auto,
		vertex/.style={circle,draw=black!100,fill=black!100, thick,
			inner sep=0pt,minimum size=1mm},scale=2.5]
		\node (O) at (-.5,0) [circle, minimum width=8pt, draw, inner
		sep=0pt, label=above:$v_1$] {}; %
		\node (A) at (0,0) [circle, dotted,minimum width=8pt, draw, inner
		sep=0pt, label=above:$v_2$] {}; %
		\draw[-latex] (O) to[bend left=20] node[above] {\small$e_2$} (A); %
		\draw[-latex] (O) to[bend left=-20] node[below] {\small{}} (A); %
		\draw[-latex] (A) to[bend left=80] node[below] {\small$e_1$} (O); %
		\node (O1) at (-1,0) [vertex,label=above:]{}; %
		\node (A1) at (.5,0) [vertex,label=above:]{}; %
		\draw[-latex] (O) to[bend left=0] node[above] {} (O1); %
		\draw[-latex] (A) to[bend left=-0] node[above] {} (A1); %
		\end{tikzpicture}}
	\caption{Using this graph $\G_1^-$ and  $\G_1^+$, we can find spectral gaps in common for all periodic magnetic potential $\tb$ acting on the polyethylene, represented by the covering graph $\G$.}\label{fig:plotspectral}
\end{figure}
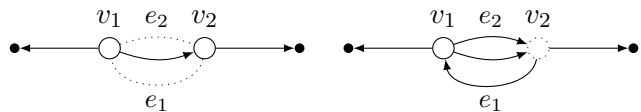

\subsection{Graphene nanoribbons}
\label{subsec:nanoribbon}
In this subsection, we will apply our method to study the example of the graphene nanoribbons (GNRs), also known as nano-graphene ribbons or nano-graphite ribbons. These are strips of graphene with semiconductive properties which are very promising
as nano-electronic devices (see, e.g., \cite{santos:09}). One of the most interesting fields of research of the nanoribbons
is the energy gaps as a function of their widths. We refer, for example, to \cite{son:06} and \cite{Han:07}. 
The GNRs repeat their geometry structure in two different ways and can be represented as $\Z$-covering graphs (see Figure~\ref{fig:ArmZigzag}).
 \begin{figure}
 	{\includegraphics[scale=.07]{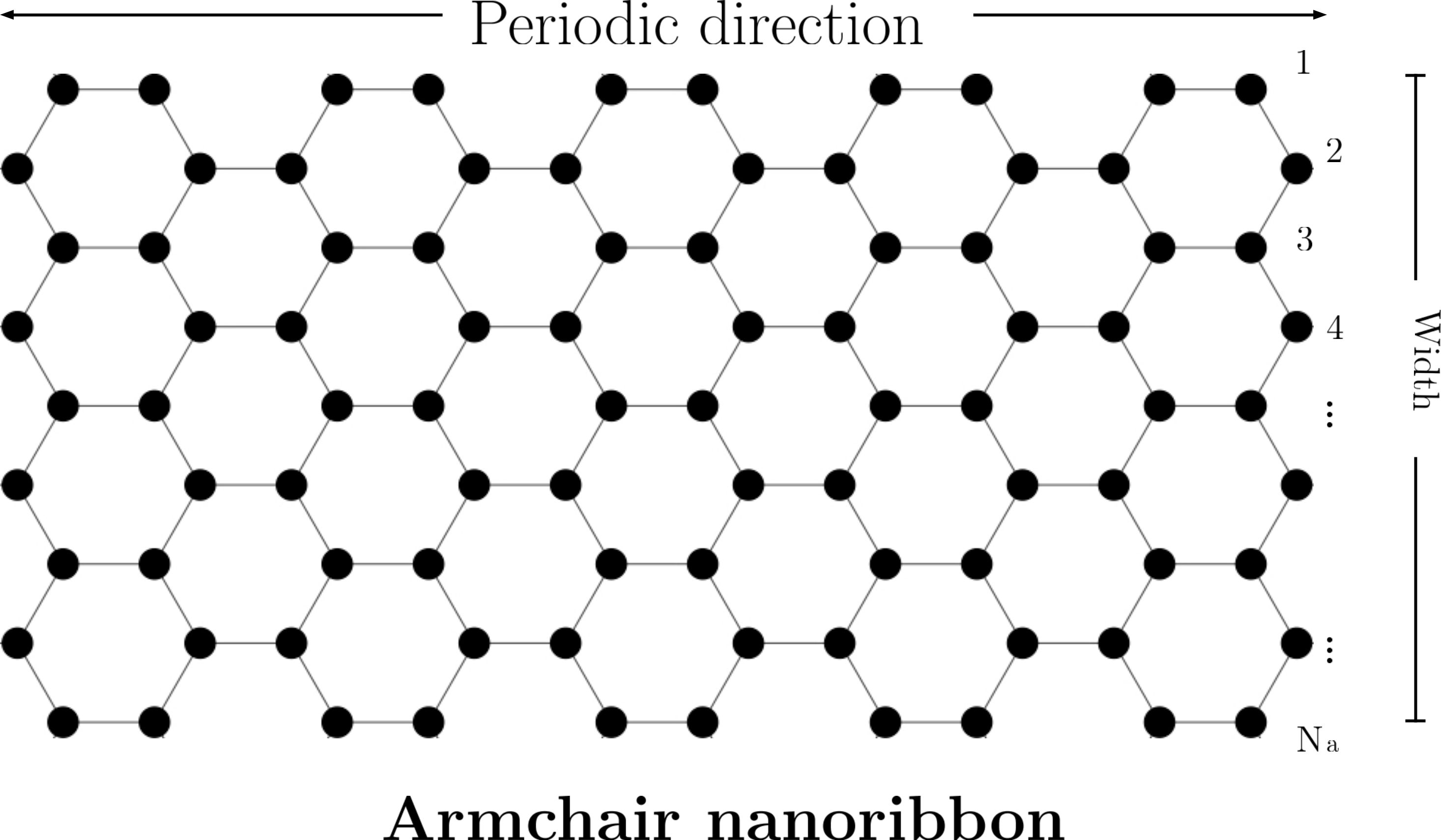}
 	\quad \includegraphics[scale=.07]{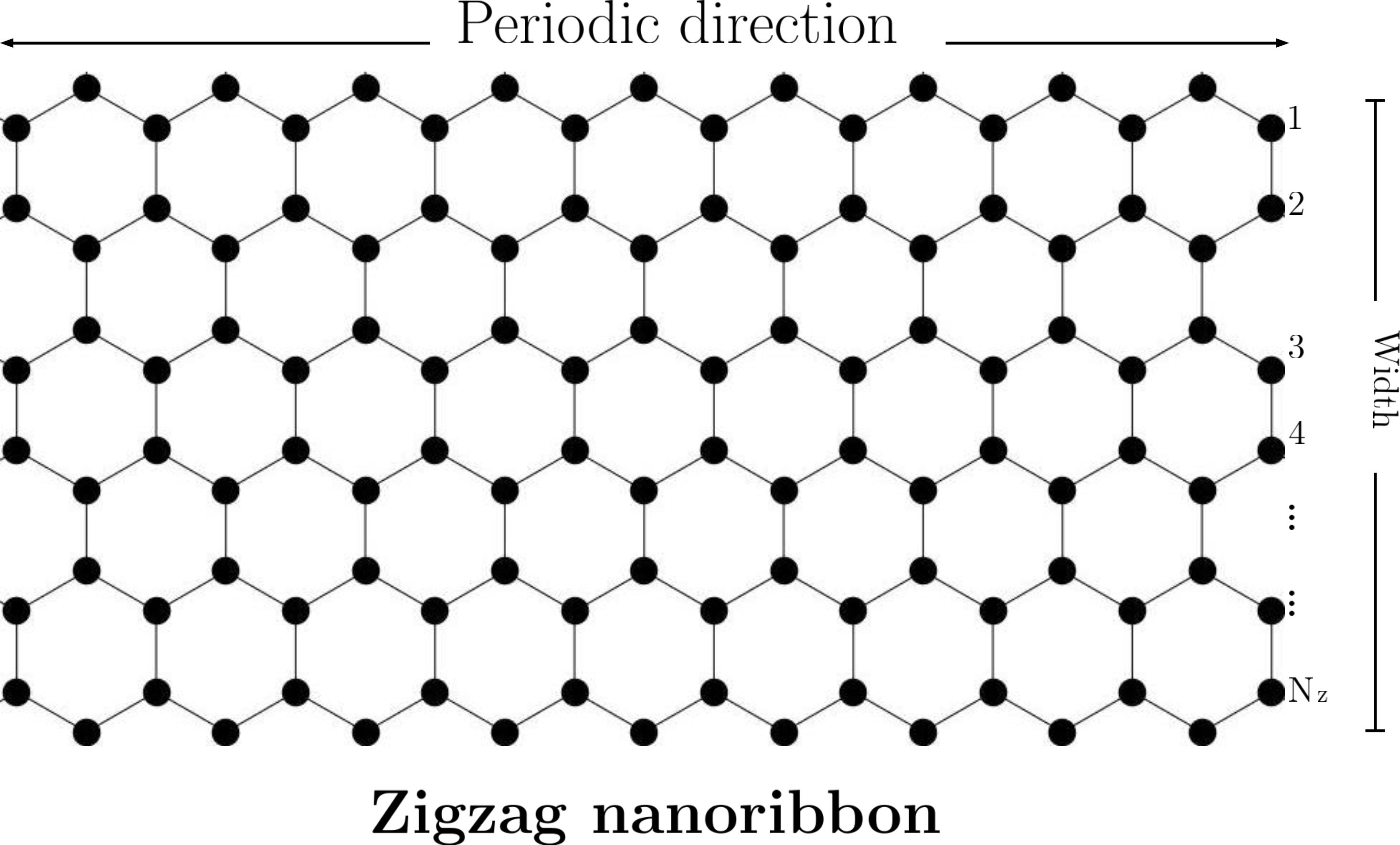}}
 	\caption{ Two structures of the graphene nanoribbons: armchair and zigzag. These structures
 	are covering graphs only in one direction. }\label{fig:ArmZigzag}
 \end{figure}
\begin{itemize}
 \item[(i)] The first variant is called armchair nanoribbon with width equal to $N_a$ and denoted as $N_a$-aGNR 
(see Figure~\ref{fig:ArmZigzag}). Consider for example 
the case of a $3$-aGNR which has similar structure as the poly-para-phenylene (PPP), one of the most important 
conductive polymers. Let $\W=(\Gt,\m)$ the MW-graph with standard weights where $\Gt$ the $\Z$-covering graph representing the $3$-aGNR and $\tb$ a constant (periodic)
magnetic potential, the idea will be use the bracketing technique to localize $\sigma(\Delta^{\Wt}_{\tb})$ and we
proceed as in the previous examples. Figure~\ref{subfig:coveringPPP} is the finite quotient graph $\G=\Gt/\Z$. 
Define in this case $E_1=\{e_1\}$ and $V_1=\{v_1\}$ so that $V_1$ is a neighborhood
of $E_1$ (see Definition \ref{def:admissible}). We construct $\W_1^+$ 
and $\W_1^-$ as before: $\G_1^+=\G-E_1$ and $\G_1^-=\G-V_1$ (cf., Figure~\ref{subfig:PPP-}). 
The weights are induced as in Definitions~\ref{deletearcs} and~\ref{deletevertices}. Using again 
the notation of the Theorem~\ref{thm:technique} and Proposition~\ref{prp:floq.mag} we obtain now a spectral localization
$J$ that depends on $\tb$. Finally, in Figure~\ref{subfig:spectraPPP} we plot the spectral bands and gaps specified by $J$ 
for the different values of the magnetic field within the interval $[0,2\pi]$. Observe that in this case
we don't have a spectral gaps common to all values of $\tb$ (as we had for the polyacetylene).

\begin{figure}[h]
	\centering

	\subcaptionbox{The quotient graph $\G$ of $3$-aGRN.\label{subfig:coveringPPP}}
	[.7\linewidth]{
		\begin{tikzpicture}[auto,
		vertex/.style={circle,draw=black!100,fill=black!100, thick,
			inner sep=0pt,minimum size=1mm},scale=.1]
		\node (1) at (5.0,-8.7) [vertex,label=above:]{}; %
		\node (2) at (-5.0,-8.7) [vertex,label=above:]{}; %
		\node (3) at (-10.0,0) [vertex,label=left:$v_1$]{}; %
		\node (4) at (-5.0,8.7) [vertex,label=above:]{}; %
		\node (5) at (5.0,8.7) [vertex,label=above:]{}; %
		\node (6) at (10.0,0) [vertex,label=above:]{}; %
		
		\draw[-latex] (1) to[bend left=0] node[above] {} (2); 	
		\draw[-latex] (2) to[bend left=0] node[above] {} (3); 
		\draw[-latex] (3) to[bend left=0] node[above] {} (4); 
		\draw[-latex] (4) to[bend left=0] node[above] {} (5); 
		\draw[-latex] (5) to[bend left=0] node[above] {} (6); 
		\draw[-latex] (6) to[bend left=0] node[above] {} (1); 	
		\draw[-latex] (3) to[bend left=0] node[above] {$e_1$} (6); 			
		\end{tikzpicture}}\\

	\subcaptionbox{The graph $\G^-$ and the graph $\G^+$ .\label{subfig:PPP-}}
	[.9\linewidth]{
			\begin{tikzpicture}[auto,
		vertex/.style={circle,draw=black!100,fill=black!100, thick,
			inner sep=0pt,minimum size=1mm},scale=.1]
		\node (1) at (5.0,-8.7) [vertex,label=above:]{}; %
		\node (2) at (-5.0,-8.7) [vertex,label=above:]{}; %
		\node (3) at (-10.0,0) [vertex,label=left:$v_1$]{}; %
		\node (4) at (-5.0,8.7) [vertex,label=above:]{}; %
		\node (5) at (5.0,8.7) [vertex,label=above:]{}; %
		\node (6) at (10.0,0) [vertex,label=above:]{}; %
		
		\draw[-latex] (1) to[bend left=0] node[above] {} (2); 	
		\draw[-latex] (2) to[bend left=0] node[above] {} (3); 
		\draw[-latex] (3) to[bend left=0] node[above] {} (4); 
		\draw[-latex] (4) to[bend left=0] node[above] {} (5); 
		\draw[-latex] (5) to[bend left=0] node[above] {} (6); 
		\draw[-latex] (6) to[bend left=0] node[above] {} (1); 	
		\draw[dotted] (3) to[bend left=0] node[above] {$e_1$} (6); 			
		\end{tikzpicture}	\quad 	\begin{tikzpicture}[auto,
		vertex/.style={circle,draw=black!100,fill=black!100, thick,
			inner sep=0pt,minimum size=1mm},scale=.1]
		\node (1) at (5.0,-8.7) [vertex,label=above:]{}; %
		\node (2) at (-5.0,-8.7) [vertex,label=above:]{}; %
		\node (3) at (-10.0,0) [circle, dotted,minimum width=8pt, draw, inner
		sep=0pt, label=left:$v_1$]{}; %
		\node (4) at (-5.0,8.7) [vertex,label=above:]{}; %
		\node (5) at (5.0,8.7) [vertex,label=above:]{}; %
		\node (6) at (10.0,0) [vertex,label=above:]{}; %
		
		\draw[-latex] (1) to[bend left=0] node[above] {} (2); 	
		\draw[-latex] (2) to[bend left=0] node[above] {} (3); 
		\draw[-latex] (3) to[bend left=0] node[above] {} (4); 
		\draw[-latex] (4) to[bend left=0] node[above] {} (5); 
		\draw[-latex] (5) to[bend left=0] node[above] {} (6); 
		\draw[-latex] (6) to[bend left=0] node[above] {} (1); 	
		\draw[-latex] (3) to[bend left=0] node[above] {$e_1$} (6); 			
		\end{tikzpicture}}\\
	\centering
	\subcaptionbox{Spectral bands and gaps as a function of the constant (periodic) magnetic potential $\tb$.\label{subfig:spectraPPP}}
	[\linewidth]{	
		\includegraphics[scale=.3]{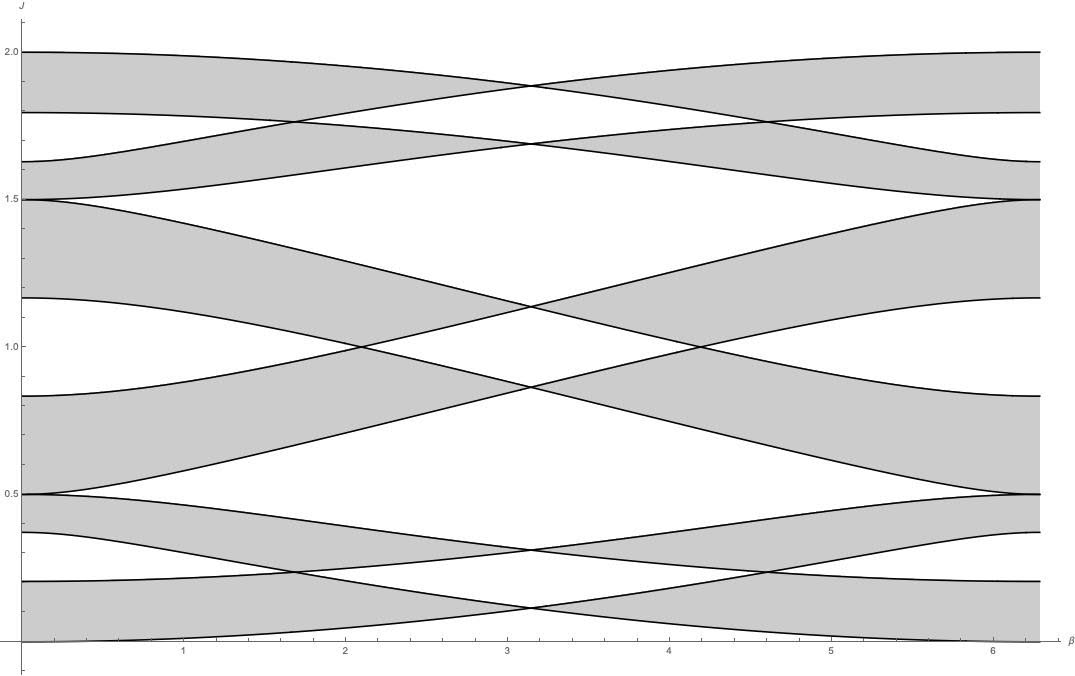}}                  	      	     	           	
	\caption{Spectral gaps of \emph{3-aGNR} for a constant magnetic potential
		$\tb=s$. Here,
		$J$ is the spectral localization of the pair $\G - \{e_1\}$ and $\G
		-\{v_1\}$, bipartitness (and interlacing) gives again the bracketing $J$ as
		localization set.}  
	\label{fig:3aGNR}
\end{figure}    
  Similar analysis could be done for any $N_a$-aGNR under the action of any periodic magnetic potential, and the bracketing technique will give good estimates
  of the intervals where the spectrum lies.
  
  Also, observe that for the combinatorial weights, we can show the existence of spectral gaps using the condition of Eq.~\ref{eq:conditioncom}
  as in \textit{Fact 2} in the polyacetylene example. We have in this case,
   \begin{equation*}
  \delta= \deg(v_1) - 2\;|[B(\subG,\Gt)]|-\lambda_1(\Delta_{\beta^-}^{W^-})>3-2-1=0.
  \end{equation*}  
  
\item[(ii)] The second variant is the so-called zigzag nanoribbon with width equal to $N_z$ are denoted as $N_z$-zGNR (see Figure~\ref{fig:ArmZigzag}). Consider $\W=(\Gt,\m)$ the MW-graph with standard weights and $\Gt$ is the graph 
 given by the zigzag nanoribbons for a fixed $N_z$, and $\tb\sim 0$ acting on $\Gt$. In this case our spectral localization method does not specify spectral gaps 
 (i.e., the spectral bands overlap). 
 The reason is that for any width $N_z$ the spectrum of the zigzag nanoribbons satisfy $\sigma(\Delta^{\Wt}_{0})=[0,2]$, i.e., in this case there are no spectral gaps.
 This fact is confirmed also by our method.
 \end{itemize}
 
\providecommand{\bysame}{\leavevmode\hbox to3em{\hrulefill}\thinspace}

\end{document}